\documentclass[12pt]{article}

\newcommand{\blind}{0}

\usepackage{preamble_things}

\begin{document}
%TC:ignore
\def\spacingset#1{\renewcommand{\baselinestretch}%
			{#1}\small\normalsize} \spacingset{1}
			
			\title{\bf Measuring sampling plan utility in post-marketing surveillance of medical products}
   \if0\blind {
   \author{Eugene Wickett$^a \text{ } (\text{corr. author; }  \href{mailto:eugenewickett@gmail.com}{eugenewickett@gmail.com})$, \\ Matthew Plumlee$^b$, Karen Smilowitz$^a$, Souly Phanouvong$^c$ and \\ Timothy Nwogu$^c$ \\
			$^a$ Industrial Engineering and Management Sciences, \\ Northwestern University,  Evanston, Illinois \\
   $^b$ Amazon, Evanston, Illinois \\
   $^c$ Promoting the Quality of Medicines Plus (PQM+) Program, \\ U.S. Pharmacopeial Convention (USP), Rockville, Maryland }
   }\fi
   \if1\blind {
   \author{Author information is purposely removed for review}
   } \fi

    \date{}
   
    \maketitle
    % \bigskip

%%% ABSTRACT 
\begin{abstract}
Ensuring product quality is critical to combating the global challenge of substandard and falsified medical products.
Post-marketing surveillance is a central quality-assurance activity in which products from consumer-facing locations are collected and tested.
Regulators in low-resource settings use post-marketing surveillance to evaluate product quality across locations and determine corrective actions.
Part of post-marketing surveillance is developing a sampling plan, which specifies where to test and the number of tests to conduct at a location.
With limited resources, it is important to base decisions on the utility of the samples tested.
We propose a Bayesian approach to generate a comprehensive utility metric for sampling plans.
This sampling plan utility integrates regulatory risk assessments with prior testing data, available supply-chain information, and valuations of regulatory objectives.
We develop an efficient method for calculating sampling plan utility.
We illustrate the value of the utility metric with a case study based on de-identified post-marketing surveillance data from a low-resource setting.
\end{abstract}


%TC:endignore

%\newpage
\spacingset{1.8}

\newpage
% INTRODUCTION
%%% NEW SECTION %%%
\section{Introduction}
% SFP OVERVIEW %
Substandard and falsified medical products (SFPs) cause substantial morbidity and mortality globally according to the World Health Organization (WHO).
Low- and middle-income countries bear the greatest SFP burdens \parencite{who2017}.
Post-marketing surveillance (PMS) is a key regulatory activity used in these countries.
The aim of PMS is to assess product quality through testing of samples at consumer-facing locations of the supply chain \parencite{nkansah2018,who2017-2}.
Regulators use analysis of testing results to determine corrective actions such as investigations, warnings and recalls.
Limited PMS budgets make the choice of sampling locations a vital decision.
Testing results from one round of PMS can feed the decision for sampling locations in the future. 
Current PMS practice does not feature a metric for the marginal value of future tests.
This paper considers how supply-chain information and prior PMS results can be leveraged to establish the value of additional tests. 
This value captures the ability of tests to improve inference of SFP prevalence that informs the deployment of corrective actions.
We provide methods for efficiently calculating this value.
Our case study demonstrates that even a simple greedy application of the value of further testing can help regulators improve the use of PMS budgets.
Generating the value of further tests thus lays a foundation for more advanced approaches to determining a set of sampling locations.

% PMS, ORIENTEERING AND THE IMPORTANCE OF HAVING A SAMPLING PLAN UTILITY %
The decisions of which locations to visit, the visiting order of these locations, and the number of tests to collect from each location is called the PMS \textit{sampling plan}.
Regulators seek sampling plans that maximize the prospective value of data produced by a sampling plan, which we refer to as the plan \textit{utility}.
% CURRENT PMS PRACTICE AND THE CHALLENGE OF ESTABLISHING UTILITIES %
Current practice for sampling plan development centers on a risk-based framework \parencite{nkansah2018}.
Regulators assess the SFP risk of locations as a function of geography, disease burden, and other characteristics.
These assessments, as well as operational requirements, desired regional variety, or other factors, then inform the construction of sampling plans.
These plans focus testing at high-risk locations.
However, this approach has limitations.
First, assessments of SFP risk may not explicitly account for important elements of prior data.
Some locations may have been tested previously, and testing at alternate locations yields different inference about upstream supply-chain conditions.
Knowing five tests from one location are as valuable as twenty tests from another location, for instance, can improve plan design.
Further, establishing a plan's utility is not straightforward.
The objective of PMS is effectively deploying corrective actions, and not necessarily maximizing SFP detection.
A utility measuring the likelihood of SFP detection is thus insufficient.
Determining future actions requires an analysis of prospective testing results prior to collection.
How these prospective results tie to regulatory decisions must be appropriately characterized to capture regulatory aims.
For example, overestimating the SFP rate at one location may be preferable to underestimating the rate at another location.
Lastly, as established in {\citet{wickett2023}}, analysis of utility cannot consider tests at different locations independently.
Supply-chain information can impact inference of product quality across all tested and non-tested locations.
In this paper, we show how a Bayesian approach can integrate regulatory objectives, assessments of SFP risk, past testing data, and available supply-chain information in a single utility measure.

% IMPORTANCE OF UNDERSTANDING SUPPLY CHAINS % 
Any location in the supply chain can be a source of SFPs \parencite{pisani2019, who2017-2,newton2011, suleman2014quality}.
Manufacturers may produce products of poor quality, distributors may substitute placebos for authentic products, or shops may store products at unsuitable temperatures.
Each supply-chain location is associated with an SFP rate, defined as the proportion of quality products passing through that location that become substandard or falsified.
Regulators continuously update estimates of these rates and deploy corrective actions accordingly.
This paper, building on the Bayesian inference approach of \citet{wickett2023}, proposes an innovative measure of the value of new samples using supply-chain information and prior testing results.
This measure can inform the locations to test that produce the data most beneficial to decision-making.

The approach to measuring sampling plan utility proposed in this paper, which grew from a partnership between operations researchers and regulatory practitioners, broadens the connection between PMS data and budgets.
Existing approaches, such as that of \citet{newton2009}, specify the sample sizes needed to form sufficient confidence intervals for SFP rates.
These intervals are compared with tolerance thresholds, where intervals exceeding thresholds trigger corrective actions.
Under constrained resources, these samples sizes may be infeasible when considered across locations.
A Bayesian, utility-based approach instead provides the value of any plan for deploying future actions, tying budgets to decisions.  
Plan utility is comprised of two main elements: the value of an estimate relative to the latent truth, called the \textit{score}, and the importance of the latent truth for regulatory objectives, called the \textit{weight}.
The score measures the deviation between an estimate and the true SFP rate, and accounts for a difference in hazard between overestimation and underestimation.
The weight captures regulator prioritization of ranges of SFP rates.
For example, SFP rates that sufficiently exceed the tolerance threshold do not need to be estimated with high accuracy: it is clear corrective actions are needed.
The product of the score and weight is the \textit{loss}, the weighted penalty for the difference between the estimate and reality when accounting for regulatory objectives.
We propose a utility equal to the expected loss across all locations.
Plans with high utility are likely to produce strong estimates where regulatory objectives are impacted.
Importantly, these valuations do not impose a sample-size requirement; any sampling plan has an associated utility.
Regulators can thus focus on increasing the value of available resources.
Furthermore, we provide methods for efficient utility estimation.
These methods can be embedded in approaches to solving the larger problem of identifying high-utility sampling plans.

% OVERVIEW OF UTILITY BENEFITS %
This approach to estimating sampling plan utility can be most beneficial in low-resource settings.
Regulators in these settings have limited capacities for the procurement, transport, and testing of samples.
However, regulators are responsible for ensuring quality for an array of medical products.
Understanding the marginal utility of sampling plans as PMS budgets change can improve operational planning.
For instance, incurring the marginal costs associated with sampling from a remote area may be inadvisable if more utility can be obtained by redirecting those resources to additional sampling from nearby regions.
Furthermore, obtaining resources for many surveillance activities depends on providing justification to funding sources, such as government bodies or nonprofit organizations.
Measurements of utility can facilitate these justifications.
Additionally, methods that are not computationally intensive are essential in low-resource settings.
We leverage the binary, pass-fail nature of much PMS testing to develop an efficient method for calculating sampling plan utility.
This method facilitates plan comparison and means of generating plans, and thus lays a foundation for more advanced approaches to determining a set of sampling locations.
A case study in Section \ref{sec:casestudy}, using de-identified PMS data from a low-resource setting, demonstrates the value of the utility metric.

% LITERATURE REVIEW %
\section{Literature review} \label{subsec:litreview}
% PMS %
Current PMS sampling plan methodology relies on assessments of SFP risk.
The risk-based approach of \citet{nkansah2018} is valuable for low-resource settings, where regulators often lack the resources necessary for statistically sufficient testing of all products, as proposed in \citet{newton2009}.
Focusing resources at high-risk areas helps address this scarcity.
However, this approach does not explicitly integrate prior testing results and available supply-chain information into sampling plan formation.
Previous work has demonstrated the value of analyzing supply chains in understanding SFP incidence.
For example, the qualitative analysis of \citet{pisani2019} showed how incentives across the supply chain interact with government policy to encourage or discourage SFPs.
\citet{wickett2023} incorporated supply-chain information through Bayesian inference to better understand SFP sources.
This inference uses testing results and supply-chain information for a fixed sampling plan.
We offer a utility-based method for developing sampling plans that builds on this inference framework while integrating the SFP-risk assessments of \citet{nkansah2018}.

% BAYESIAN EXPERIMENTAL DESIGN %
This method adds to a literature of settings using Bayesian experimental design.
A design is set of input parameter choices for an experiment; see a review in \citet{chaloner1995}.
Bayesian experimental design suits settings involving binary, count, or categorical data where data collection is expensive, such as PMS.
The loss determines the cost of a decision made from an experiment with respect to the latent truth.
\citeauthor{chaloner1995} stressed the importance of constructing a loss matching experimental aims.
In Section \ref{sec:designUtil}, we propose such a loss for PMS.

Bayesian experimental design has been applied to a variety of problems using different losses or utilities, where a utility is an integration of the loss over a distribution of key parameters.
\citet{lookman2019}, for example, used an expected improvement utility for materials discovery. 
% MAYBE REMOVE THESE SENTENCES %
This utility fits settings where the goal is to maximize a response: the best designs balance uncertainty with the prospect for response maximization.
\citeauthor{lookman2019} aimed to identify materials maximizing the level of different physical attributes.
This contrasts with PMS, where the objective is an effective deployment of corrective actions, which may not correspond with heightened SFP detection.
% end MAYBE REMOVE %
\citet{phillips2021} used a measure of the distance between the prior and posterior, called the Shannon information, for key parameters in nuclear physics experiments.
The aim of these experiments is to tighten the feasible constraint range for modeled parameters or constants.
This relates to PMS decision thresholds: corrective actions are deployed once the credible interval for a location's associated SFP rate exceeds the threshold.
\citet{morgan2020} studied parameters of carbon capture systems, and discussed how loss specifications might change as experimentation continues and objectives evolve.
Initial experiments may target uncertainty across input ranges, while later experiments focus data collection at inputs deemed crucial for program objectives.
Similar dynamism exists in PMS settings, as discussed in Section \ref{subsec:defLoss}: initial iterations provide broad assessments of SFP rates while later iterations focus on distributing corrective actions.

% AUDITING %
Improving the understanding of quality issues through PMS can draw on previous work in auditing.
Auditing problems usually consider game-theoretic settings where a buyer's inspection policies elicit reactions by suppliers: see
\citet{dawande2020} for a review of auditing problems in supply chains.
\citet{zhang2022,huang2022} considered violations of social responsibility, e.g., abusive labor practices, at suppliers for a buyer who could be impacted by such violations.
The buyer decides where to audit and how much effort to expend.
These papers identified optimal auditing policies for different contexts within two- or three-echelon supply chains.
This setting is similar to PMS in that the buyer chooses locations to audit and corrective actions to deploy.
However, the focus is modeling incentives and violations as a function of these audits, rather than improving detection.

% ORIENTEERING %
Identifying a utility-maximizing sampling plan can be modeled as a variant of the orienteering problem \parencite{tsiligirides1984}.
Orienteering problems seek a utility-maximizing tour through a subset of locations, subject to budget constraints.
The utility of an orienteering problem is usually additive across visited locations, representing some reward for visiting locations.
However, the value of testing data in PMS is not independent across locations. 
Locations share a connected supply chain, and thus information about one location affects inference at another location.
This article considers calculating a sampling plan utility that accounts for inference of quality problems across the supply chain.
In ongoing work, we develop approaches to solving a formulation for sampling plan optimization based on the orienteering problem.
This formulation requires a way to evaluate prospective sampling plans.

% NETWORK DETECTION %
Network detection problems use sensors throughout a network to either discover anomalies or learn network-based attributes.
\citet{anzoom2021} reviewed approaches for disrupting illicit networks.
They labeled approaches as having either a network or a supply-chain view.
Network views consider general collections of nodes and links; supply-chain views distinguish nodes by their role.
Our work carries the supply-chain view: nodes are distinguished in that only consumer-facing nodes are tested.
We target learning SFP rates for improved decision-making; the work reviewed in \citeauthor{anzoom2021} generally sought to maximize detection.
For instance, \citet{triepels2018} built a Bayesian network model for detecting document fraud in international shipping.
Bayesian networks are systems of structured variable dependencies; these dependencies yield conditional probabilities for variables that require significantly less data to estimate than amassing all possible combinations of variables.
Bayesian networks are beneficial in settings with variable structure like international shipping, where different types of goods are frequently shipped together to similar places.
\citeauthor{triepels2018} formed posterior likelihoods of the presence or omission of different goods on shipping itineraries; the presence of low-likelihood goods or the omission of high-likelihood goods, relative to established thresholds, triggered inspection flags.
Similar decision thresholds exist for PMS in low-resource settings, where challenges exceed resources and some SFP risk must often be tolerated.

Sensor location problems consider how best to deploy limited measurement capacity on a network to learn an important network attribute.
These problems parallel test deployment at a subset of nodes in PMS.
For example, \citet{yang1998} and \citet{bianco1996} considered sensor locations on a traffic network for learning the origin-destination matrix governing the network.
Origin-destination matrices indicate the traffic between any two nodes.
\citet{yang1998} identified the sensor location subsets of a given size that minimized the worst-case error.
\citet{bianco1996} determined lower and upper bounds on the minimal cost of sensor locations that ensure learning traffic flows.

% PAPER SUMMARY %
The rest of this paper is as follows.
Section \ref{sec:samplingDesigns} describes key elements of sampling plans in PMS and defines sampling plan utility.
Section \ref{sec:designUtil} proposes methods of building and applying a loss for SFP-rate estimates that tie key regulatory objectives to sampling plan utility.
Section \ref{sec:casestudy} uses these methods with a 2021 set of de-identified PMS data and demonstrates the application of utility to the generation of favorable sampling plans.
Section \ref{sec:discussion} summarizes the advantages of plan utility.

% SAMPLING PLANS IN SUPPLY CHAINS
%%% NEW SECTION %%%
% What is a sampling plan for our considered context? What data can we use for sampling plan formation? %
\section{Sampling plans in supply chains} \label{sec:samplingDesigns}
In this section, we describe sampling plans and a proposed utility metric that ties data collection to the effective deployment of corrective actions.

%%% NEW SUBSECTION %%%
% How does PMS work? %
\subsection{Overview of post-marketing surveillance} \label{subsec:pmsIter}

PMS is an iterative process of data collection and analysis, centered on WHO's framework of prevention, detection, and response \parencite{who2017-2}.
This process is usually conducted for a single class of medical products.
For instance, elevated malaria incidence, or interest in malarial parasite resistance \parencite{chikowe2015}, may motivate a PMS iteration for anti-malarial products.
The availability of procurement, transportation, and testing resources is subject to funding and approval procedures. 
PMS data are often produced in batches as these procedures are completed.

Figure \ref{fig:PMSiteration} illustrates the iterative PMS process.
Step 1 is forming the sampling plan.
Current practice centers on the Medicines Risk-based Surveillance framework and web tool  \parencite{nkansah2018}.
Regulators assess locations' SFP risk as a function of geography, disease burden, past corrective actions, and other inputs.
These assessments inform choices of where to sample.
In our experience in low-resource settings, a plan may comprise collection of a few dozen to a few thousand samples for a single PMS iteration.

%%% FIGURE 
\begin{figure}[hbt]
    \centering
    \begin{localsize}{11}

\usetikzlibrary{shapes.geometric}
% Define block styles
\begin{tikzpicture}[
squarednode/.style={rectangle, draw=gray!60, fill=gray!25, very thick,minimum height = 0.9cm,text width=7.1cm, align=left},
% squarednode2/.style={rectangle, draw=green!60, fill=green!25, very thick,minimum height = 0.9cm,text width=3.5cm, align=left},
textnode/.style={rectangle, draw=gray!50, fill=magenta!0, dashed, thick,  minimum height = 1.0cm, text width=5.6cm, align=left},
node distance=10mm
]

\coordinate (textorig) at (3.3, -1.8);
\coordinate (arrowend) at (0.0, -1.8);

%Nodes
\node[squarednode](b1){\footnotesize 1. \textbf{Form sampling plan}: Determine how many samples to obtain from each consumer-facing location, using assessments and available data};
\node[squarednode](b2)[right= 1.8cm of b1] {\footnotesize 2. \textbf{Procure samples}: Visit locations and obtain samples according to the sampling plan};
\node[squarednode](b3)[below = 2.5cm of b2]  {\footnotesize 3. \textbf{Conduct testing}: Evaluate samples' quality with chosen testing method};
\node[squarednode](b4)[left = 1.8cm of b3]{\footnotesize 4. \textbf{Analyze results}: Use testing data to form estimates of SFP rates at locations in the supply chain and apply corrective actions};
\node[textnode](b5)  at (textorig) {\footnotesize \textit{This paper}: Use supply-chain information and prior testing results to define the utility of a given sampling plan};
% \node[squarednode2](b6)[left = 1.8cm of b4]{\footnotesize 4a. \textbf{Apply corrective actions}: Use analysis to apply targeted or system-wide corrective actions};

\draw[line width=0.6mm,->,dashed,gray!50] (b5) -- node[below=3ex]{} (arrowend) ;

%Lines
\draw[line width=0.7mm, ->] (b1) to  node[pos=0.85,left]{} (b2);
\draw[line width=0.7mm, ->] (b2) to  node[pos=0.85,left]{} (b3);
\draw[line width=0.7mm, ->] (b3) to  node[pos=0.85,left]{} (b4);
\draw[line width=0.7mm, ->] (b4) to  node[pos=0.85,left]{} (b1);

\end{tikzpicture}

\end{localsize}
\caption[PMS iterative process]{{PMS iterative process for a designated product class.}}
\label{fig:PMSiteration}
\end{figure}
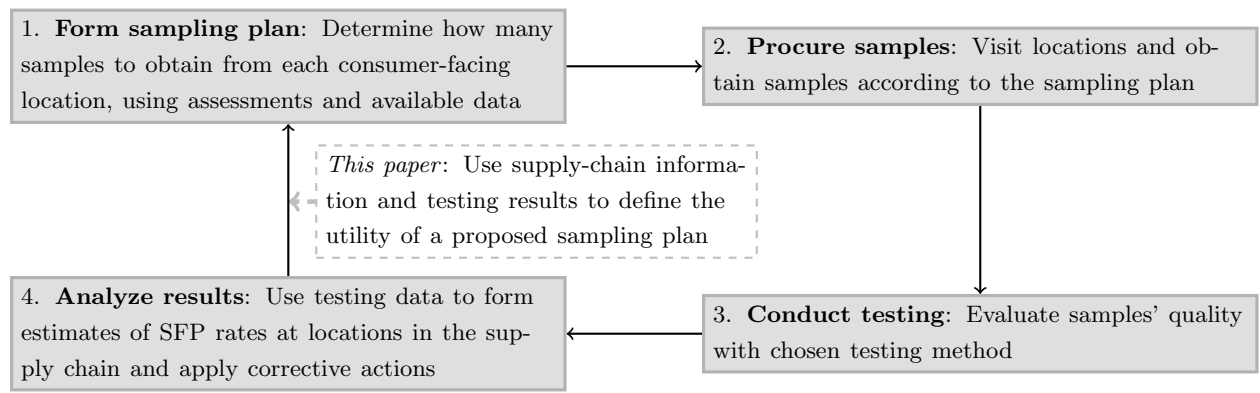
%%% END FIGURE

Step 2 is sample procurement, where locations designated by the plan are visited.
Regulators may procure any available product of the designated class, or regulators may choose among products that reflect different supply-chain paths.
Each product has an associated supply-chain path that is a list of locations visited by a product.
Frequently, the full path is unavailable and only partial information such as a manufacturer or distributor label can be accessed.
With increased use of track-and-trace systems, most or all of a sample's supply-chain path will become known \parencite{rotunno2014}.
We assume that each product has a label for the visited location and one upstream location, and products are randomly sampled across paths.

Step 3 is testing samples against product registration specifications.
\citet{kovacs2014} described a variety of testing methods used in PMS, ranging from visual inspection to X-ray diffraction analysis.
Some testing methods provide an array of measurements for the sample, such as the amount of active pharmaceutical ingredient or the amount of impurities.
We consider binary pass/fail determinants of sample quality, as reflected in common PMS data sets like the Medicines Quality Database \parencite{mqd2021}.
These pass/fail measurements align with the screening second level of the three-level inspection approach of \citet{pribluda2014}, where the first level is physical and visual inspection and the third level is a complete pharmacopeial assessment.
While we assume perfect accuracy of pass/fail detection, imperfect testing accuracy can be easily incorporated.
Each record in a PMS data set then contains a binary test result and associated metadata, including available supply-chain information such as procurement location or manufacturer.
Practically, there may be sample attrition between procurement and testing: samples may be damaged or lost, or testing materials may be degraded or unavailable.
However, for our modeling we assume there is no attrition between procurement and testing.
Thus, ``samples'' and ``tests'' will often be used interchangeably in this paper.

Step 4 is analysis.
Regulators use PMS with two related, yet distinct, objectives: a \textit{classification} objective or an \textit{assessment} objective.
The classification objective categorizes locations as either high or low SFP risk.  
Corrective actions are deployed at high-risk locations.
The assessment objective seeks accurate estimation of the SFP rates associated with locations.
Classification may fit a context of deploying a limited number of process or operational inspections.
Assessment may fit a context where regulators seek inference of SFP rates in the supply chain, such as a setting of studying a new or previously untested class of products.
This inference may then inform broad corrective actions, such as changes to registration procedures, or the decision to conduct further surveillance.
After analysis, regulators apply any corrective actions and iterate to Step 1.

This paper extends the analysis stage by considering how testing results affect sampling plan design for the next PMS iteration.
Our aim is to devise a utility metric capturing a plan's ability to improve decision-making.
For instance, allocating samples to a province with a low associated SFP rate may be helpful if that province is connected to key distributors for which there is little information.

%%% NEW SUBSECTION %%%
\subsection{Inference in post-marketing surveillance}
\label{subsec:inferenceinPMS}
Regulators in low-resource settings generally have access to limited supply-chain information at sample collection, such as one or two labels from upstream supply locations.
The example in Figure \ref{fig:twoSCechelons} illustrates the approach of \citet{wickett2023} to leverage limited supply-chain information to improve inference of SFP rates.
A \textit{node} is a location within an echelon. 
A sample's \textit{trace} is the supply-chain information associated with that sample.
\textit{Test nodes} are consumer-facing locations, and \textit{supply nodes} are locations from one upstream echelon.
The number below each test node $\tnPoint$, $\allocVec_\tnPoint$, is the number of samples obtained from $\tnPoint$.
The fraction on each adjacent edge is the number of detected SFPs over the number of samples from that trace.
For example, 8 of 11 samples obtained from Test Node 4 were sourced from Supply Node 1, and 2 of these 8 samples were SFPs.
The principal challenge is that a product can become substandard or falsified at any supply-chain location, but testing in PMS is only done at test nodes.

%%% FIGURE 
\begin{figure}[ht]
    \centering
    \resizebox{.75\linewidth}{!}{ % \begin{localsize}{11}

\usetikzlibrary{shapes.geometric}
% Define block styles
\begin{tikzpicture}[
squarednode/.style={rectangle, draw=blue!60, fill=blue!25, very thick,minimum height = 0.9cm,text width=2.5cm, align=center},
roundnode/.style={ellipse, draw=magenta!60, fill=magenta!25, very thick,  minimum height = 1.5cm, text width=1.1cm, align=center},
node distance=10mm
]
%Nodes
\node[squarednode](i1){\small Supply Node 1};
\node[squarednode](i2)[right = 4.4cm of i1]  {\small Supply Node 2};
\node[roundnode, label=below:{\large $\allocVec_1=12$}](o1)[below left=4.2cm and -0.4cm of i1] { Test Node 1};
\node[roundnode, label=below:{\large $\allocVec_2=3$}](o2)[below right=4.2cm and -0.4cm of i1] { Test Node 2};
\node[roundnode, label=below:{\large $\allocVec_3=7$}](o3)[below left=4.2cm and -0.4cm of i2] { Test Node 3};
\node[roundnode, label=below:{\large $\allocVec_4=11$}](o4)[below right=4.2cm and -0.4cm of i2] { Test Node 4};

%Lines
\draw[thick,-] (i1) to  node[pos=0.85, left] {{\Large $\frac{3}{7}$}} (o1) ;
\draw[thick,-] (i1) to  node[pos=0.88, right] {\Large{$\frac{0}{0}$}} (o2);
\draw[thick,-] (i1) to  node[pos=0.98, above] {\Large{$\frac{0}{3}$}} (o3);
\draw[thick,-] (i1) to  node[pos=0.98, above] {\Large{$\frac{2}{8}$}} (o4);

\draw[thick,-] (i2) to node[pos=0.98, above] {\Large{$\frac{1}{5}$}} (o1);
\draw[thick,-] (i2) to node[pos=0.98, above] {\Large{$\frac{0}{3}$}} (o2);
\draw[thick,-] (i2) to node[pos=0.88, left] {\Large{$\frac{0}{4}$}} (o3);
\draw[thick,-] (i2) to node[pos=0.87, right] {\Large{$\frac{1}{3}$}} (o4);;

% Make a legend
% \matrix [draw,below left] at (current bounding box.south west) {
%   \node [roundnode,label=right:Foo] {}; \\
%   \node [squarednode,label=right:Bar] {}; \\
% };
\node[draw,left = 1.5cm of o1,minimum height=1.1cm,minimum width=2.5cm] (lgd) {$\frac{\text{\small positive tests}}{\text{\small total samples}}$};
\draw[dashed,->] (lgd) -- ++ (0.0,0.55) |-  (-3.2,-3.9) ;

% Other possible echelons
\node[rectangle, thick, draw=lightgray, minimum width = 103mm,text width=25mm, align=center,  minimum height = 0.9cm, anchor=west, dash pattern=on 0.2cm off 0.1cm] (r2) at (-1.5,-1.3) {\small\textit{Distribution}\vspace{18pt}};
\node[rectangle, thick, draw=lightgray, minimum width = 103mm,text width=45mm, align=center, minimum height = 1.1cm, anchor=west, dash pattern=on 0.2cm off 0.1cm] (r1) at (-1.5,1.3) {\small\textit{Upstream sourcing}};

\draw[thick, -, color=lightgray,dash pattern=on 0.2cm off 0.1cm] (r1) to (i1);
\draw[thick, -, color=lightgray,dash pattern=on 0.2cm off 0.1cm] (r1) to (i2);

\end{tikzpicture}

% \end{localsize}
}
\caption[Example supply chain]{{Example of two echelons of a supply chain with PMS testing results.
Data are only collected from the lower echelon of test nodes.
``Upstream sourcing'' and ``Distribution'' signify echelons for which information is unavailable.}}
\label{fig:twoSCechelons}
\end{figure}
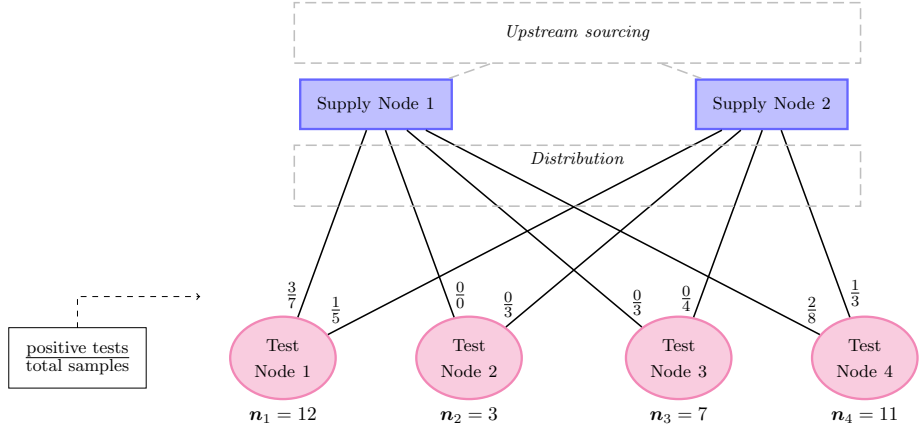
%%% END FIGURE

A PMS data set contains testing results and traces for each sample.
Let $\tnSet$ be the set of test nodes and let $\snSet$ be the set of supply nodes.
Regulators collect and test samples from test nodes in $\tnSet$. 
All products pass through one supply node and one test node.
A set of $\numTestsNext$ tests yields a set of binary results $\bm{\testResultpoint}=\{\testResultpoint_1,\ldots,\testResultpoint_{\numTestsNext}\}$: 
a ``1'' denotes SFP detection and ``0'' denotes no SFP detection.
Traces are stored as $(\bm{\tnPoint},\bm{\snPoint})$, where $\bm{\tnPoint}=\{\bm\tnPoint_1,\ldots,\bm\tnPoint_{\numTestsNext}\}$ is the set of test-node labels and $\bm{\snPoint}=\{\bm\snPoint_1,\ldots,\bm\snPoint_{\numTestsNext}\}$ is the set of supply-node labels.
The overall data set is $\dataSetNext=(\bm{\testResultpoint},\bm{\tnPoint},\bm{\snPoint})$.

The approach of \citet{wickett2023} seeks to infer the SFP rates associated with test nodes and supply nodes using PMS data $\dataSetNext$.
Test-node SFP rates are stored in vector $\tnRates=(\tnRates_1,\ldots,\tnRates_{\tnNum})\in(0,1)^{|\tnSet|}$ and supply-node SFP rates are stored in vector $\snRates=(\snRates_1,\ldots,\snRates_{\snNum})\in(0,1)^{|\snSet|}$.
A test can be classified as an SFP due to causes associated with either the test node or supply node (or further upstream).
The consolidated SFP rate of a test with trace $(\tnPoint,\snPoint)$ accounts for SFP causes associated with either test node $\tnPoint$ or supply node $\snPoint$ (or further upstream), and is stated as
\begin{equation} \label{eq:consolSFPProb}
    \consolSFPfunc_{\tnPoint\snPoint}(\SFPrateSet) = \tnRates_{\tnPoint} + (1 - \tnRates_{\tnPoint}) \snRates_{\snPoint} \,.
\end{equation}
The log-likelihood of $\tnsnRates=(\SFPrateSet)$ under data $\dataSetNext$ is
\begin{equation} \label{eq:loglike}
    \ell(\tnsnRates|\dataSetNext) = \sum_{\tnPoint\in\tnSet} \sum_{\snPoint\in\snSet} 
    \bigg[ \log[{\consolInaccFunc}_{\tnPoint\snPoint}(\tnsnRates)] \numTestsNextMat_{\tnPoint\snPoint} + \log[ 1-{\consolInaccFunc}_{\tnPoint\snPoint}(\tnsnRates)] (\numTestsNextMat_{\tnPoint\snPoint} - \testResultNextMat_{\tnPoint\snPoint}) \bigg] \,,
\end{equation}
where $\numTestsNextMat_{\tnPoint\snPoint}$ is the number of tests in $\dataSetNext$ from trace $(\tnPoint,\snPoint)$ and $\testResultNextMat_{\tnPoint\snPoint}$ is the number of observed SFPs.
\citet{wickett2023} showed this log-likelihood is unidentifiable: for any SFP rates $\tnsnRates$, there exist alternate rates $\tnsnRates'\neq\tnsnRates$ with equal likelihood; i.e., there exists $\tnsnRates'$ such that $\ell(\tnsnRates'|\dataSetNext)=\ell(\tnsnRates|\dataSetNext)$.
Thus, when resources are limited to testing at consumer-facing locations, as in common PMS, testing data alone cannot yield exact conclusions regarding SFP rates.

This unidentifiability is mitigated through a Bayesian approach.
Let $\priorFunc(\tnsnRates)$ be a prior on $\tnsnRates$ that is designated through regulator assessments of SFP risk.
Multiplying $\priorFunc(\tnsnRates)$ with the likelihood under data $\dataSetNext$, $\exp\big[\ell(\tnsnRates | \dataSetNext)\big]$, is then proportional to a posterior, i.e.,
\begin{equation} \label{eq:posterior_inferencepaper}
 \postFunc(\tnsnRates | \dataSetNext) \propto \exp \big[\ell(\tnsnRates | \dataSetNext)\big] \priorFunc(\tnsnRates)\,.
\end{equation}
We can infer SFP rates associated with nodes at both tested and upstream locations with Markov Chain Monte Carlo (MCMC) draws from this posterior.
Credible intervals formed from these draws are compared with designated thresholds to inform the deployment of corrective actions.

%%% NEW SUBSECTION %%%
\subsection{SFP-rate inference and sampling plan utility} \label{subsec:planUtil}
The utility of a sampling plan is directly tied to SFP-rate inference.
As with {\citet{wickett2023}},
we consider two echelons of a larger, more complex supply chain in analyzing sampling plans.
Sampling plan $\allocVec\in(\posIntSet)^{|\tnSet|}$ is an allocation of samples where $\allocVec_{\tnPoint}$ is the allocation to test node $\tnPoint$ in $\tnSet$.
Each PMS iteration has a budget of $\numTestsNext$ samples; thus, $\sum_{\tnPoint\in\tnSet}\allocVec_{\tnPoint}\leq\numTestsNext$.
Consistent with PMS practice, we express the budget by the number of samples instead of the cost to procure samples.

Bayesian experimental design uses existing testing data and priors on SFP rates to explore how the value of inference might be improved through different sampling plans in the next PMS iteration \parencite{chaloner1995}.
Inference improvement is defined by a value function capturing what constitutes a good estimate of SFP rates.
Proper characterization of this value function is key to Bayesian experimental design.

Our aim is to determine the utility of new tests.
Let $\dataSet=(\testResultVec,\tnLabelVec,\snLabelVec)$ be the set of existing PMS data from previous iterations and let $\dataSetNext=(\testResultVecNext,\bm{\tnPoint},\bm{\snPoint})$ be a new data set.
The supply node associated with each test in $\dataSetNext$ is randomly drawn according to historic sourcing records.
The estimate of SFP rates, $\tnsnRatesEst(\dataSetNext,\dataSet)$, is a function of existing data $\dataSet$ as well as new data $\dataSetNext$ once sampling plan $\allocVec$ is deployed.
Let $\estVal\big[\tnsnRatesEst(\dataSetNext,\dataSet),\tnsnRates\big]$ be a loss that describes the penalty accrued as $\tnsnRatesEst$ differs from $\tnsnRates$, a vector of true SFP rates.
$\estVal$ comprises of a score, $S \big[\tnsnRatesEst(\dataSetNext,\dataSet),\tnsnRates\big]$, and a weight, $W(\tnsnRates)$, that together capture how estimates impact regulatory objectives.
The score reflects the difference between an estimate and the truth; the weight defines the significance of the truth for regulatory goals.
The loss of an estimate formed using $\dataSetNext$ sums the score and weight for each node:
\begin{equation} \label{eq:pmsloss}
\estVal  \big[  \tnsnRatesEst(\dataSetNext,\dataSet),\tnsnRates \big] = \sum_{\tnsnPoint\in\tnsnSet} \score \big[\tnsnRatesEst_\tnsnPoint(\dataSetNext,\dataSet),\tnsnRates_\tnsnPoint\big] \weight(\tnsnRates_\tnsnPoint)\,.
\end{equation}
The utility of plan $\allocVec$ is then the expected loss reduction from new data:
\begin{equation} \label{eq:designutil}
    \util (\allocVec) = \expecOp \bigg[ \estVal\big[\tnsnRatesEst(\dataSet),\tnsnRates\big]\bigg] - \expecOp\bigg[\estVal\big[\tnsnRatesEst(\dataSetNext|\allocVec,\dataSet),\tnsnRates\big] \bigg]\,.
\end{equation}
A good plan samples from test nodes with a high likelihood of yielding data that reduce estimate loss.
Plan utility provides a single metric for plan comparison and can inform evaluation of PMS budgets: analysis of the change in utility with respect to data set size may indicate budgets that do not provide high marginal value after a certain size.

% INTERPRETING UTILITY IN POST-MARKETINGSURVEILLANCE
%%% NEW SECTION %%%
\section{Interpreting utility in post-marketing surveillance} \label{sec:designUtil}
Calculating sampling plan utility requires a loss formulation that aligns SFP-rate estimates with regulatory objectives.
Section \ref{subsec:defLoss} presents loss formulation that accounts for an estimate's accuracy (score) and its significance for regulatory decision-making (weight).
Section \ref{subsec:bayesestimators} defines Bayes estimates that minimize the expected loss under new data.
Section \ref{subsec:designUtilExample} illustrates the application of sampling plan utility.
Section \ref{subsec:calcprep} presents an efficient method for utility estimation.

%%% NEW SUBSECTION %%%
\subsection{Loss construction for SFP-rate estimates} \label{subsec:defLoss}
The loss captures both an estimate's accuracy and the importance of accuracy at different SFP rates.
General functions exist for measuring estimate quality, e.g., an estimate's squared error.
However, decision-makers often desire tailored functions with properties better suited to objectives \parencite{hennig2007}, particularly in Bayesian experimental design \parencite{chaloner1995}.

Characterization of the score and weight in (\ref{eq:pmsloss}) depends on the regulatory objective(either a classification or an assessment objective).
Section  \ref{subsubsec:objClass} describes the classification objective, where regulators aim to classify nodes as either high or low SFP risk.
Section \ref{subsubsec:objAssess} describes the assessment objective, where regulators aim to form useful inference of SFP rates.
The loss construction we propose is not burdensome: the regulator specifies only one parameter for the classification objective, and two parameters for the assessment objective.  

%%% NEW SUBSUBSECTION %%%
\subsubsection{Classification objective} \label{subsubsec:objClass}
A classification objective looks to correctly identify nodes that are significant sources of SFPs.
Identification is determined through classification threshold $\threshold$.
Nodes with estimated SFP rates above $\threshold$ are classified as significant SFP sources, and corrective action is taken.

\paragraph{Classification score}
Let 
\[
\classFunc(\tnsnRates_\tnsnPoint) = \begin{cases} 1 \text{ if } \tnsnRates_\tnsnPoint \geq \threshold
\text{ and} \\
0 \text{ if } \tnsnRates_\tnsnPoint < \threshold
\end{cases}
\]
be a classification for SFP rate $\tnsnRates_\tnsnPoint$ at node $\tnsnPoint$ with respect to threshold $\threshold$.
The classification score for estimate $\tnsnRatesEst_\tnsnPoint$ of latent SFP rate $\tnsnRates_\tnsnPoint$ is
\begin{equation} \label{eq:scoreClassify} S_1\left(\tnsnRatesEst_\tnsnPoint,\tnsnRates_\tnsnPoint\right) =  \Big[C(\tnsnRatesEst_\tnsnPoint) - C(\tnsnRates_\tnsnPoint)\Big]^{+}+ \underEstWt\Big[C(\tnsnRates_\tnsnPoint) -C(\tnsnRatesEst_\tnsnPoint)\Big]^{+}
\end{equation}
for any $\tnsnPoint\in\tnsnSet$, where $(x)^+=\max(0,x)$.
The first term captures overestimation error and the second term captures underestimation error, magnified by parameter $\underEstWt$.
Reducing the score means better classifying node SFP rates correctly as below or above $\threshold$.
% In our experience, $\threshold$ can range as high as 20\%, depending on public health considerations and the underlying SFP prevalence across the supply chain.
% For example, suppose a node has a true SFP rate of 30\% in a setting where $\threshold$ is $20\%$. 
% An estimate of 80\% may seem a poor estimate when compared with an estimate of 15\%, but the 80\% estimate is preferable in that in correctly classifies the node as a significant SFP source.

User-determined parameter $\underEstWt$ is the penalty of underestimation relative to overestimation.
Overestimation means potentially mis-allocating corrective actions; underestimation means potentially neglecting a significant SFP source.
Increasing $\underEstWt$ favors estimates that result in classifying more nodes as SFP sources.
One can consider $\underEstWt$ as the trade-off between corrective actions at high-SFP and low-SFP nodes: each accurate action requires potentially $\underEstWt$ inaccurate actions.
Over-allocating resources like corrective actions is preferred to neglecting issues in many public-health situations.
For instance, \citet{wachtel2010} and \citet{olivares2008} identified a $\underEstWt$ around 2 for balancing overstaffing and effectiveness in an operating room under uncertain operating times.
With PMS, the number of corrective actions deployed depends on the number of high-SFP nodes.
If regulators expect few such nodes, $\underEstWt$ can be set high: acting against true high-SFP nodes means possibly acting against many low-SFP nodes.
If regulators expect many high-SFP nodes, $\underEstWt$ can alternatively be set low.
This trade-off may depend on the corrective action class being considered: more punitive actions may require higher classification accuracy.

%%% NEW SUBSUBSECTION %%%
\subsubsection{Assessment objective} \label{subsubsec:objAssess}
An assessment objective measures estimate accuracy and the significance of the latent SFP rate for regulatory decision-making.

\paragraph{Assessment score}
The assessment score measures the deviation between estimate $\tnsnRatesEst_\tnsnPoint$ and latent rate $\tnsnRates_\tnsnPoint$:
\begin{equation} \label{eq:scoreAssess} S_2(\tnsnRatesEst_\tnsnPoint,\tnsnRates_\tnsnPoint) = (\tnsnRatesEst_\tnsnPoint - \tnsnRates_\tnsnPoint)^{+}+ \underEstWt(\tnsnRates_\tnsnPoint - \tnsnRatesEst_\tnsnPoint)^{+}\,,
\end{equation}
for each $\tnsnPoint\in\tnsnSet$.
The assessment score mirrors the classification score in balancing overestimation and underestimation through $\underEstWt$, but differs in considering the absolute difference between estimate and true rates, instead of the classification difference.

\paragraph{Assessment weight} 
The weight reflects the relative importance of different ranges of SFP rates for decision-making.
Uncertainty for rates near the decision threshold is more hazardous than rates away from the threshold.
The assessment weight for latent rate $\tnsnRates_\tnsnPoint$ is
\begin{equation} \label{eq:riskCheck}
    W(\tnsnRates_\tnsnPoint) = 1-\tnsnRates_\tnsnPoint\left[\checkRiskSlope-\mathbbm{1}\{\tnsnRates_\tnsnPoint<\threshold\}\left(1-\frac{\threshold}{\tnsnRates_\tnsnPoint}\right) \right] \,,
\end{equation}
for each $\tnsnPoint\in\tnsnSet$.
Parameter $\checkRiskSlope\in[0,1]$ designates the slope of the weight away from threshold $\threshold$.
The weight is highest at $\threshold$ independent of $\checkRiskSlope$.
Figure \ref{fig:rateweights} plots weight values for different choices of $\checkRiskSlope$ for $\threshold=30\%$.
For rate $\tnsnRates_\tnsnPoint<\threshold$, the weight is $1-(\checkRiskSlope-1)\tnsnRates_\tnsnPoint-\threshold$; for rate $\tnsnRates_\tnsnPoint\geq\threshold$, the weight is $1-\checkRiskSlope\tnsnRates_\tnsnPoint$.
Values for $\checkRiskSlope$ below $0.5$ produce a reduced slope for SFP rates higher than $\threshold$ and an increased slope for rates lower than $\threshold$: values above $\threshold$ are weighted higher than values below $\threshold$.
Values for $\checkRiskSlope$ above $0.5$ produce the opposite effect: values below $\threshold$ are instead weighted higher.

%%% WEIGHT FIGURE %%%
\begin{figure}[ht]
    % \centering
    \begin{center}
        \includegraphics[width=0.65\textwidth]{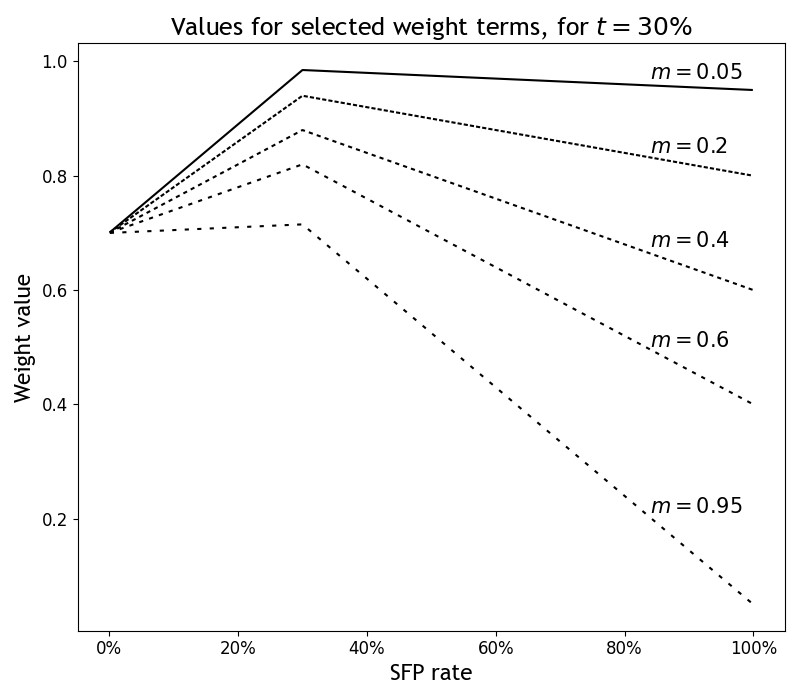}
    \end{center}
    \caption[Illustration of assessment weights]{Illustration of weights for $\threshold=30\%$ and $\checkRiskSlope\in\{0.05,0.2,0.4,0.6,0.95\}$.}
    \label{fig:rateweights}
\end{figure}
%%% END FIGURE %%%

The weight is a key element of loss specification for the assessment objective, though results do not seem overly sensitive to exact weight values.
Estimation accuracy for rates near $\threshold$ is more important than accuracy for rates far from $\threshold $: poor estimation of rates far from $\threshold$ is less likely to have an impact on regulatory decisions. 
For example, distinguishing a 5\% SFP rate from a 20\% SFP rate may be more important than distinguishing a 70\% rate from an 85\% rate, although the difference within each pair of rates is the same. 
Corrective actions will be deployed in the latter case for either rate, while the deployment of corrective actions in the former case may depend on the latent reality.
A value of $\checkRiskSlope$ between 0.6 and 0.95 meets the common PMS aim of de-emphasizing high latent rates while focusing utility at rates near $\threshold$.
This unequal weighting of rates above or below $\threshold$ through a single parameter, $\checkRiskSlope$, is the appeal of using (\ref{eq:riskCheck}) over other weight forms, such as a parabolic weighting.
Where $\threshold$ is not well established, using a high $\threshold$ appropriate to the regulatory setting, in combination with an $\checkRiskSlope$ above 0.9, places approximately equal emphasis on rates below the highest reasonable $\threshold$ while de-emphasizing very high rates.

%%% NEW SUBSUBSECTION
\subsection{Forming Bayes estimates from PMS data} \label{subsec:bayesestimators}
This section introduces Bayes estimates that integrate PMS data with the score and weight values to evaluate sampling plan utility. 

Bayes estimates are decisions that minimize expected loss \parencite{Ghosh2006}.
In PMS, these are SFP-rate estimates that best integrate regulatory goals and testing data.
The loss is a measure of the quality of estimate $\tnsnRatesEst$: the loss incorporates how estimates will be used, the penalty of overestimation versus underestimation, and the weighted relevance of the true SFP rate being estimated.
The true SFP rates are unknown; these rates are inferred via $p(\tnsnRates|\dataSetNext)$, the posterior using available data.
The Bayes estimate minimizes the expected loss under $p(\tnsnRates|\dataSetNext)$:
\begin{equation}
    \tnsnRatesEst( \dataSetNext ) = \argmin\limits_{\Dot{\tnsnRates}} \mathbbm{E}_{\tnsnRates|\dataSetNext}\estVal(\Dot{\tnsnRates},\tnsnRates)\,,
\end{equation}
where $\Dot{\tnsnRates}$ is a vector with length $|\tnsnSet|$.
The Bayes estimate can be found by derivation, by using standard optimization software, or by selecting the member of a set of MCMC draws from $p(\tnsnRates|\dataSetNext)$ that minimizes the average loss across all MCMC draws.
Derived Bayes estimates for the loss constructions of Section \ref{subsec:defLoss} are described in Appendix \ref{SuppMat:bayesEstimates}.

%%% NEW SUBSECTION %%%
\subsection{Benefits of the utility metrics} \label{subsec:designUtilExample}

Sampling plan utility sharpens plan selection in at least three ways, which we illustrate through an example.
First, plan utility provides a single metric for how well a plan integrates objectives with data; this metric facilitates plan comparison.
Second, plans can be compared for given budgets, allowing evaluation of the trade-off between utility and budget.
Third, plans meeting different operational constraints can be compared, yielding a utility cost for these constraints.

Recall the example of Figure \ref{fig:twoSCechelons}, which consists of 4 test nodes, 2 supply nodes, 33 tests and 7 positive results.
Suppose the goal is to ascertain the best sampling plan for the next $\numTestsNext$ tests.
Regulators want to explore SFP prevalence; thus, they apply the assessment objective.
The number and results of tests already collected along each test node-supply node trace are respectively provided as $\numTestsMat$ and $\testResultMat$, where row $i$ denotes Test Node $i$ and column $j$ denotes Supply Node $j$:
\begingroup
\[ 
\renewcommand{\arraystretch}{0.5}
\numTestsMat=\left[ \begin{array}{cc}
7 & 5 \\
0 & 3 \\
3 & 4 \\
8 & 3
\end{array} \right],
\testResultMat=
\left[ \begin{array}{cc}
3 & 1 \\
0 & 0 \\
0 & 0 \\
2 & 1
\end{array} \right]\,.
\]
\endgroup
We see from $\testResultMat$ that SFPs have been detected at Test Nodes 1 and 4, sourced from Supply Nodes 1 and 2.
Inference from these data is not strong enough to take corrective actions: one cannot conclude if SFP issues occurred due to test nodes or locations upstream in the supply chain.
Therefore, regulators plan to collect more tests.
Suppose three plans are under consideration.
The proportion of tests allocated to each test node under any plan can be expressed in vector form:
\[ 
\allocVec^{\text{LT}}=\left[ \begin{array}{cc}
0 \hspace{5pt} %\\
\numTestsNext \hspace{5pt} %\\
0 \hspace{5pt} %\\ 
0
\end{array} \right]^{\intercal}, \hspace{3pt}
\allocVec^{\text{U}}=\left[ \begin{array}{cc}
\frac{\numTestsNext}{4} \hspace{5pt} %\\
\frac{\numTestsNext}{4} \hspace{5pt} %\\
\frac{\numTestsNext}{4} \hspace{5pt} %\\ 
\frac{\numTestsNext}{4}
\end{array} \right]^{\intercal}, \hspace{3pt}
\allocVec^{\text{HS}}= \left[ \begin{array}{cc}
\frac{\numTestsNext}{2} \hspace{5pt} %\\
0 \hspace{5pt} %\\
0 \hspace{5pt} %\\
\frac{\numTestsNext}{2}
\end{array} \right]^{\intercal}\,.
\]
These plans meet different concerns.
Plan $\allocVec^{\text{LT}}$, a ``least tested'' plan, allocates all tests to Test Node 2, with the least prior tests.
Plan $\allocVec^{\text{U}}$, a ``uniform'' plan, allocates tests to all nodes equally.
Plan $\allocVec^{\text{HS}}$, a ``highest SFPs'' plan, evenly allocates tests to Test Nodes 1 and 4, with prior detected SFPs.

%%% FIGURE %%%
\begin{figure}
    \centering
    \includegraphics[width=0.48\textwidth]{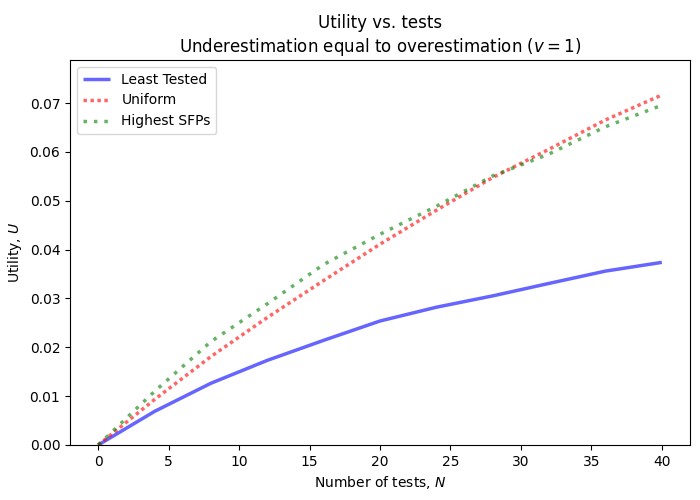} 
    \hfill
    \includegraphics[width=0.48\textwidth]{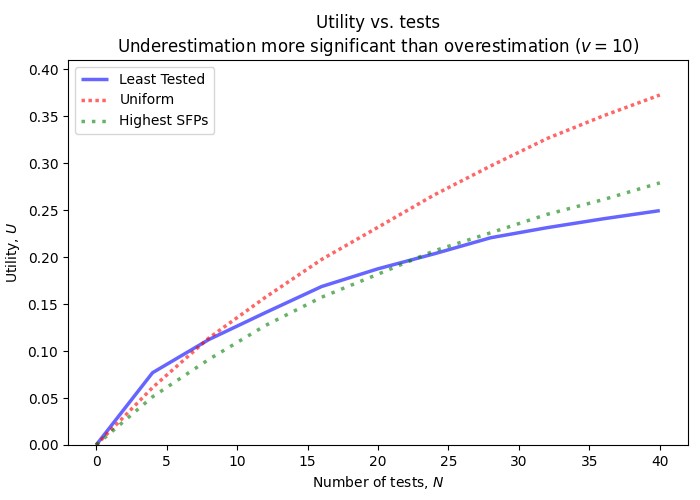} 
    \caption[Illustration of utility versus budget for example plans]{{Illustration of utility for different plans as the size of the next set of tests, $\numTestsNext$, increases in intervals of four tests. 
    }}
    \label{fig:exampleSCmargutil}
\end{figure}
%%% END FIGURE %%%

Suppose regulators designate the loss with a score that equally considers overestimation and underestimation  ($\underEstWt=1$) and a weight using threshold $\threshold=20\%$ and slope $\checkRiskSlope=0.6$.
Figure \ref{fig:exampleSCmargutil} (left) shows the utility for each plan as the total number of tests, $\numTestsNext$, increases in intervals of four tests.
This plot shows a few key aspects of plan utility.
First, there is a diminishing return on additional testing: less is learned through continued testing.
Second, the preferred plan can depend on budget: the Highest SFPs plan provides higher utility until the budget exceeds 25 tests, where the Uniform plan becomes preferable.
For low budgets, estimates of node SFP rates are best improved by testing where SFPs have been previously observed.
Third, the utility of the Highest SFPs plan is close to the utility of the Uniform plan. 
As the Highest SFPs plan requires visiting only two of the four nodes, this plan may still be preferable even for high budgets.

Changes to loss parameters underscore the need to accurately characterize the loss.
Suppose the score parameter is changed to $\underEstWt=10$, meaning underestimation is $10$ times worse than overestimation.
Figure \ref{fig:exampleSCmargutil} (right) shows the resulting utility evaluations as $N$ increases.
The difference in scales for the $y$-axis of each figure results from different loss specifications. 
The key for analysis is determining which plan provides the highest utility as compared with other plans---not identifying the highest absolute utility value across loss specifications.
Changing $\underEstWt$ from 1 to 10 makes the Highest SFPs plan the worst plan from 0 to 20 tests. 
This makes sense when considering there are only 3 tests from Test Node 2: the likely range of SFP rates for this node is wide, so there is a strong underestimation risk.
Interestingly, the Least Tested plan provides the most initial value for the ``next test,'' though this plan's advantage dissipates after about 5 tests. 
Additionally, the Uniform plan dominates the other two plans once $\numTestsNext$ exceeds around 10 tests: the risk of underestimating at Test Node 2 is outweighed by the risk of underestimating at other nodes. 
A balanced approach best ensures no significant SFP sources are missed.

% OVERVIEW OF EFFICIENT ESTIMATION AND IMPORTANCE SAMPLING %
\subsection{Efficient estimation of sampling plan utility}
\label{subsec:calcprep}
Moving beyond the simple example, a utility-based approach to sampling plan development requires methods for efficient utility estimation.
MCMC integration is standard for estimating expressions like sampling plan utility \parencite{kleingesse2019}.
However, the computational resources needed for standard MCMC integration are prohibitive for many low-resource contexts: calculating the utility for a single sampling plan requires about 16 hours with an Intel Core 1.6GHz processor.
Regulators with limited computational resources cannot then analyze the benefit of different plans, which is a key practical benefit of considering plan utility.

A faster utility approximation is found by avoiding the repeated MCMC sampling of standard MCMC integration when estimating a loss and instead coupling a single set of MCMC draws with a collection of simulated data sets, detailed in Appendix \ref{SuppMat:calcUtil}.
The posterior for the expected loss is expressed as $p(\tnsnRates|\dataSet,\dataSetNext) p(\dataSetNext|\dataSet)$ \parencite{chaloner1995}.
Manipulation of this posterior produces an equivalent posterior that separates SFP-rate MCMC draws under known data from the binomial likelihood of an SFP-rate vector under a simulated data set:
\[
 p(\tnsnRates|\dataSet,\dataSetNext) p(\dataSetNext|\dataSet) = p(\tnsnRates|\dataSet) p(\dataSetNext|\tnsnRates) \,.
 \]
This likelihood can be efficiently calculated for each SFP-rate draw and simulated data set.
An efficient approximation for the expected loss is generated through a weighted sum of loss realizations for each SFP-rate draw in a single MCMC set.
The resulting utility approximation requires around 10 minutes for settings similar in size to the case study in Section \ref{sec:casestudy}.

This efficient approximation of utility is subject to bias under settings with sufficiently large supply chains or numbers of tests.
A small subset of the MCMC-generated set of SFP-rate draws acquire an outsized impact on the resulting loss estimate.
Bias stems from identifying Bayes minimizers that are overfit to this subset: the resulting loss estimate is too low.
Increasing the number of SFP-rate MCMC draws mitigates this bias.
However, the time required to generate and manipulate a sufficiently large set of draws can approach the time needed for standard MCMC integration.
We propose an importance-sampling approach for efficiently mitigating the bias of large systems, which is shown to work well in the case study.
The first step of this approach is identifying  the important region of the space of SFP rates for a given supply chain and sampling plan.
This region is then oversampled.
The resulting set of MCMC draws, coupled with their importance weights, is used for efficient loss estimation.

Software package 
\if1\blind
{\textit{[software removed for blind review]}}
\fi 
\if0\blind
{\texttt{logistigate}}, available on GitHub \parencite{logistigate2021},
\fi 
provides methods for utility estimation, including efficient approximation and the importance sampling approach.

% CASE STUDY
%%% NEW SECTION %%%
\section{Integrating utility into sampling plan generation} \label{sec:casestudy}
The development of a new utility metric is part of a larger collaboration with a medical product regulatory agency implementing PMS in a low-resource setting.  
Using a de-identified case study with data from this agency, we illustrate how this metric can assist plan evaluation.

%%% NEW SUBSECTION %%%
\subsection{Case study background}
The agency conducted the first PMS iteration for a class of maternal and child health products in 2021.
Products had not yet been broadly tested at consumer-facing locations.
The agency tested 177 samples from 4 provinces near the capital, of which 62 were SFPs (a 35\% rate).
Samples featured 13 unique manufacturer labels.

The sampling plan for the first PMS iteration was developed using the risk-based approach of \citet{nkansah2018}.
The agency will conduct a second PMS iteration for the same products to form estimates of SFP rates throughout the supply chain (assessment objective).
The first iteration allocated samples to the four provinces nearest the central laboratory.
This met operational constraints, including time and transportation costs.
The usual approach for the next iteration would be to re-apply these assessments to formation of the new plan.

The operational constraints for the next iteration, including the PMS budget, are largely dependent on variable funding sources that are not yet established.
Thus, this case study considers a range of possible budgets. 
There are also two considered implementation settings: an \emph{existing} setting and an \emph{all-provinces} setting.
The existing setting only considers allocations to the provinces of the first iteration.
The all-provinces setting considers every province.

%%% NEW SUBSECTION %%%
\subsection{Preparations for applying utility}
Section \ref{subsubsec:contextprep} describes the preparation necessary for aligning utility with the regulatory context.
Section \ref{subsubsec:greedyheuristic} describes a greedy heuristic for plan formation.

%%% NEW SUBSUBSECTION %%%
\subsubsection{Capturing the regulatory context in utility calculations}
\label{subsubsec:contextprep}
Construction of the loss was informed by collaboration with the regulatory agency. 
An SFP-rate threshold of $\threshold=15\%$ approximated the agency's threshold of serious concern.
The agency indicated a higher tolerance for overestimation versus underestimation of SFP rates: about five false alarms per true detection ($\underEstWt=5$) justifies the use of correction resources.
We use an assessment weight with $\checkRiskSlope=0.6$ to de-emphasize high versus low SFP rates, while prioritizing rates near $\threshold$.
Sensitivity analysis shows generally small but expected changes under different choices
\if0\blind{\parencite{wickettdissertation}.}\fi
\if1\blind{\parencite{wickettdissertationblind}.}\fi

We estimated sourcing probabilities for test nodes visited in the first PMS iteration through the traces observed in that iteration, as was done in \citet{wickett2023}.
Elicitation of priors on SFP rates uses an adaptation of the risk assessment process of \citet{nkansah2018}.
This process places nodes in one of seven SFP-risk categories.
Our elicitation methodology and details on sourcing probability estimation are described further in Appendix \ref{SuppMat:caseStudyPrep}.

%%% NEW SUBSUBSECTION %%%
\subsubsection{Plan formation through a greedy heuristic}
\label{subsubsec:greedyheuristic}
Efficient utility calculation opens valuable avenues of analysis for regulators.
However, identifying a utility-maximizing plan remains a challenge: calculating the utility of each possible plan is combinatorically impractical.
While our ongoing research is exploring this challenge, this section shows that even a simple heuristic that leverages efficient calculation and marginal utility can form sampling plans with high utility across a budgetary range.

%%% ALGORITHM FOR GREEDY HEURISTIC %%%
\begin{algorithm}
    \small
    \DontPrintSemicolon
        \KwInput{$\tnSet$, $\numTestsNext$, $\util(\cdot)$ }
        \KwOutput{Set of test allocations $\big\{\allocVec^{(i)}:i\in\{1,\dots,\numTestsNext\}\big\}$ for each possible budget $i$}
        Initialize $\allocVec^{(0)}_\tnPoint:=0$ for $\tnPoint\in\tnSet$.
        \\
        \For{$i$ in $\{1,\dots,\numTestsNext\}$}{
            \For{$\tnPoint\in\tnSet$}{
                $F_\tnPoint= \util(\allocVec^{(i-1)}+\bm{e}_\tnPoint) - \util(\allocVec^{(i-1)})$.    
            }
            Set $a'=\argmax_a F_a$.\\
            \For{$\tnPoint\in\tnSet$}{
                \begin{fleqn}[0pt]
                \begin{align*}
                    \allocVec^{(i)}_{\tnPoint}=
                    \begin{cases}
                        \allocVec^{(i-1)}_{\tnPoint}+1 &\text{ if } \tnPoint=\tnPoint' \\
                        \allocVec^{(i-1)}_{\tnPoint} &\text{ otherwise.}
                    \end{cases}
                \end{align*}
                \end{fleqn}
            }
        }
         Return $\{\allocVec^{(i)}:i\in\{1,\dots,\numTestsNext\}\}$.
    \caption{Greedy heuristic for sampling plan generation with varying budgets}
    \label{alg:greedyallocmain}
\end{algorithm}
%%% END GREEDY HEURISTIC %%%

This heuristic finds the best allocation for the $i^{\text{th}}$ test, given an allocation of the first $(i-1)$ tests.
Algorithm \ref{alg:greedyallocmain} takes as inputs a set of test nodes ($\tnSet$), a number of tests ($\numTestsNext$), and a function yielding an allocation's utility $\big(U(\cdot)\big)$.
The output is a set of allocation vectors $\{\allocVec^{(i)}\}_{i=1}^{\numTestsNext}$, where each vector has length $|\tnSet|$, for all budgets $i\in\{1,\dots,\numTestsNext\}$.
The algorithm proceeds as follows.

\paragraph{Initialize}
In Line 1, initialize the allocation to each test node as zero: $\allocVec^{(0)}_\tnPoint:=0$ for each $\tnPoint\in\tnSet$.

\paragraph{Add node with the best marginal utility increase}
Loop through budgets from 1 to $\numTestsNext$ tests.
Estimate the utility of a test from each available test node in Lines 3 and 4.
Marginal utility $F_{\tnPoint}$ is the difference between the utility of the last budget, $U\big(\allocVec^{(i-1)}\big)$, and the utility of the last budget plus one test at node $\tnPoint$; $\bm{e}_{\tnPoint}$ denotes a basis vector with a 1 at position $\tnPoint$. 
In Lines 5 through 7, allocate a test to the test node with the highest marginal utility.

Algorithm \ref{alg:greedyallocmain} uses $\big(|\tnSet|\times\numTestsNext\big)$ utility evaluations and leverages the sampling plans under successive budgets to identify the next best allocation.
Intervals of $\numTestsNext$, e.g., batches of five tests, can be used at each step to reduce the number of estimates.
The smallest possible interval justifying transport to a region may be an appropriate interval of $\numTestsNext$.
Figure \ref{fig:greedyallocation} illustrates three iterations of the heuristic.

%%% FIGURE 
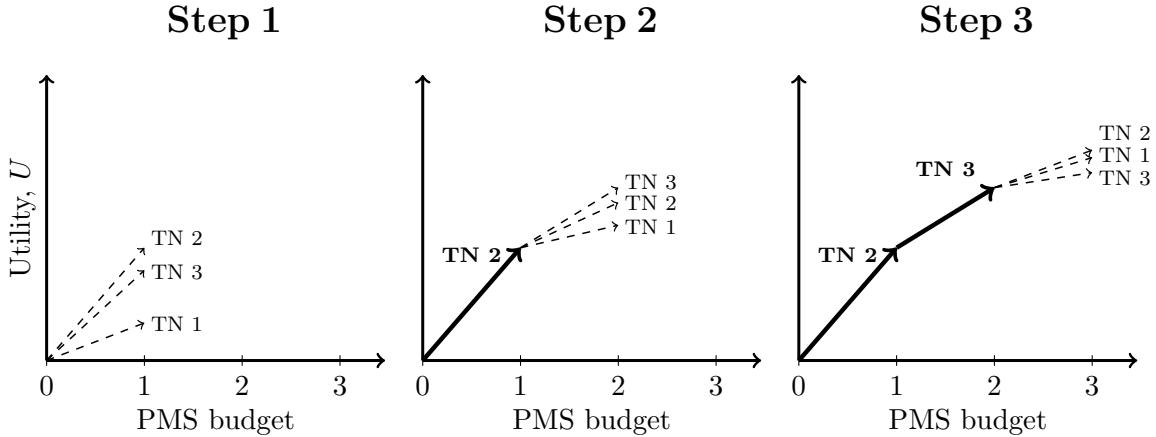
\begin{figure}[h]
    \centering
    \begin{tikzpicture}
\newcommand\Xa{0}
\newcommand\Xb{5.0}
\newcommand\Xc{10.0}
\newcommand\Xlen{4.5}
\newcommand\Ylen{3.8}
\coordinate (origa) at (\Xa,0);
\coordinate (origb) at (\Xb,0);
\coordinate (origc) at (\Xc,0);
\coordinate (Yaxa) at (\Xa,\Ylen);
\coordinate (Yaxb) at (\Xb,\Ylen);
\coordinate (Yaxc) at (\Xc,\Ylen);
\coordinate (Xaxa) at (\Xa+\Xlen,0);
\coordinate (Xaxb) at (\Xb+\Xlen,0);
\coordinate (Xaxc) at (\Xc+\Xlen,0);

% draw x , y lines
\draw[line width=0.4mm,->] (origa) -- node[below=3ex] {PMS budget} (Xaxa) ;
\draw[line width=0.4mm,->] (origb) -- node[below=3ex] {PMS budget} (Xaxb) ;
\draw[line width=0.4mm,->] (origc) -- node[below=3ex] {PMS budget} (Xaxc) ;
\draw[line width=0.4mm,->] (origa) -- node[midway,above,rotate=90]  {Utility, $U$} (Yaxa) ;
\draw[line width=0.4mm,->] (origb) -- node[]  {} (Yaxb) ;
\draw[line width=0.4mm,->] (origc) -- node[]  {} (Yaxc) ;
% draw x ,y ticks
\foreach \i [count=\j from 0] in {0.25, 0.5, 0.75, 1}
{
   \draw (\j*1.3,2pt) -- ++ (0,-4pt) node[below] {$\j$};
   \draw (\Xb+\j*1.3,2pt) -- ++ (0,-4pt) node[below] {$\j$};
   \draw (\Xc+\j*1.3,2pt) -- ++ (0,-4pt) node[below] {$\j$};
   % \draw (2pt,\j*0.8+1) -- ++ (-4pt,0) node[left]  {$\i$};
}
\node[text width=3cm] at (3.1,4.5) 
    {\Large{\textbf{Step 1}}};
\node[text width=3cm] at (3.1+\Xb,4.5) 
    {\Large{\textbf{Step 2}}};
\node[text width=3cm] at (3.1+\Xc,4.5) 
    {\Large{\textbf{Step 3}}};
% draw arrows for utility improvement
% STEP 1
\newcommand\Xticka{\Xa+1.3}
\draw[line width=0.2mm,->,dashed] (origa) -- node[above right=0.1mm and 6mm]  {\scriptsize{TN 1}} (\Xticka,0.5) ;
\draw[line width=0.2mm,->,dashed] (origa) -- node[above right=6.5mm and 6mm]  {\scriptsize{TN 2}} (\Xticka,1.5) ;
\draw[line width=0.2mm,->,dashed] (origa) -- node[above right=3.5mm and 6mm]  {\scriptsize{TN 3}} (\Xticka,1.2) ;
% STEP 2
\newcommand\Xtickb{\Xb+1.3*2}
\coordinate (step2junc) at (\Xticka+\Xb,1.5);
\draw[line width=0.6mm,->] (origb) -- node[above=4mm ]  {\scriptsize{\textbf{TN 2}}} (step2junc) ;

\draw[line width=0.2mm,->,dashed] (step2junc) -- node[above right=-1mm and 6mm]  {\scriptsize{TN 1}} (\Xtickb,1.5+0.3) ;
\draw[line width=0.2mm,->,dashed] (step2junc) -- node[above right=0.6mm and 6mm]  {\scriptsize{TN 2}} (\Xtickb,1.5+0.6) ;
\draw[line width=0.2mm,->,dashed] (step2junc) -- node[above right=2.5mm and 6mm]  {\scriptsize{TN 3}} (\Xtickb,1.5+0.8) ;
% STEP 3
\newcommand\Xtickc{\Xc+1.3*3}
\coordinate (step3junc1) at (\Xticka+\Xc,1.5);
\coordinate (step3junc2) at (\Xtickc-1.3,2.3);
\draw[line width=0.6mm,->] (origc) -- node[above=4mm ]  {\scriptsize{\textbf{TN 2}}} (step3junc1) ;
\draw[line width=0.6mm,->] (step3junc1) -- node[above=4mm ]  {\scriptsize{\textbf{TN 3}}} (step3junc2) ;

\draw[line width=0.2mm,->,dashed] (step3junc2) -- node[above right=0.0mm and 6mm]  {\scriptsize{TN 1}} (\Xtickc,2.3+0.4) ;
\draw[line width=0.2mm,->,dashed] (step3junc2) -- node[above right=2.5mm and 6mm]  {\scriptsize{TN 2}} (\Xtickc,2.3+0.5) ;
\draw[line width=0.2mm,->,dashed] (step3junc2) -- node[above right=-2mm and 6mm]  {\scriptsize{TN 3}} (\Xtickc,2.3+0.2) ;
\end{tikzpicture}
\caption[Greedy heuristic example]{{{Greedy heuristic example with three test nodes (TNs).}}}
\label{fig:greedyallocation}
\end{figure}
%%% END FIGURE

%%% NEW SUBSECTION %%%
\subsection{Benefits of considering plan utility} \label{subsec:CSinitial}
The PMS budget for the next iteration may be unknown, so we consider a budgetary range up to 400 samples, more than twice the first budget.
We evaluate in intervals of 10 samples: 17 samples was the smallest test-node allocation in the first iteration, and a finer precision is unnecessary for the agency's planning.
Utility-informed allocations derive from the greedy heuristic.
Each of the utility-informed allocations in the existing and all-provinces settings is compared with two different allocations: a \textit{uniform} allocation where samples are distributed evenly across all test nodes, and a 
\textit{fixed} allocation where samples are distributed according to the sampling plan of the first iteration.

%%% BEGIN CASE-STUDY FIGURE %%%
\begin{figure}
    \centering
    \subfloat[Utility-informed sample allocation in the existing setting. \label{subfig:CS_fam_alloc}]{%
    \includegraphics[width=0.47\textwidth]{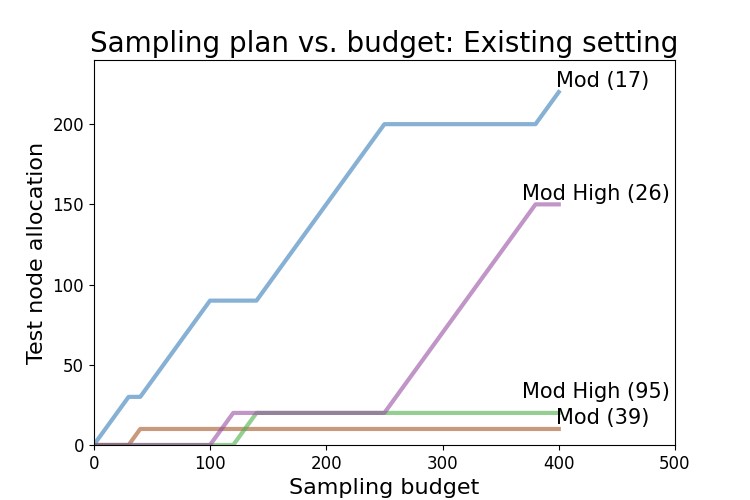} }
    \hfill
    \subfloat[Utility-informed sample allocation in the all-provinces setting. \label{subfig:CS_expl_1}]{%
    \includegraphics[width=0.47\textwidth]{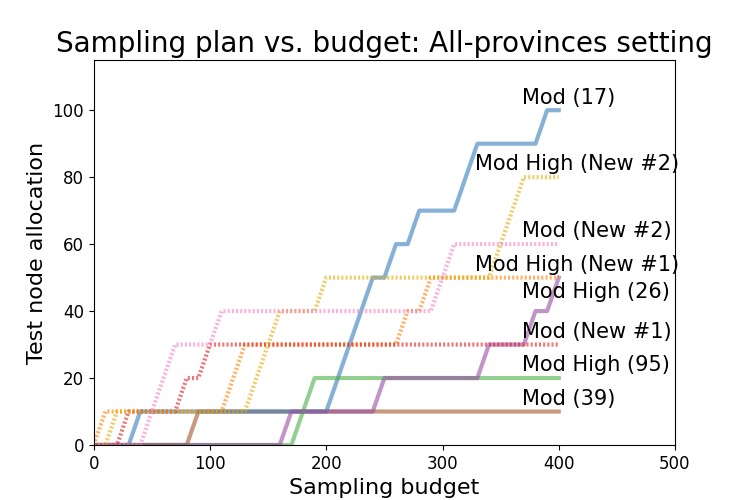} }
    \\
    \subfloat[95\% confidence intervals for utility of the utility-informed, uniform, and fixed allocations in the existing setting.\label{subfig:CS_fam_utilcompare}]{%
    \includegraphics[width=0.47\textwidth]{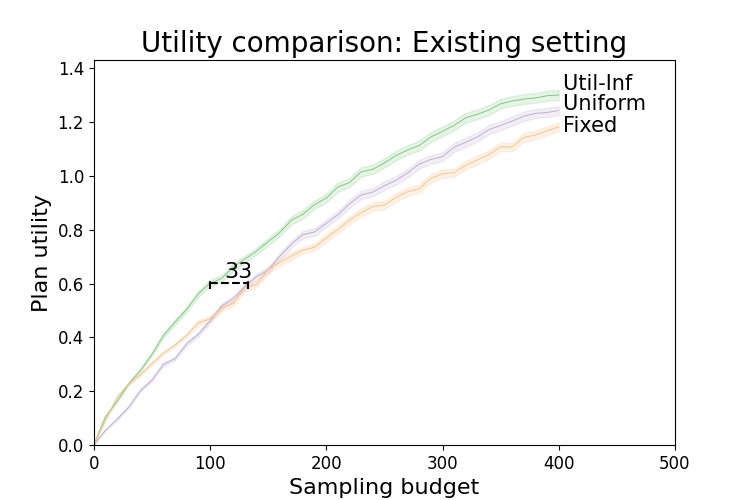} }
    \hfill
    \subfloat[95\% confidence intervals for utility of the utility-informed, uniform, and fixed allocations in the all-provinces setting.\label{subfig:CS_expl_2}]{%
    \includegraphics[width=0.47\textwidth]{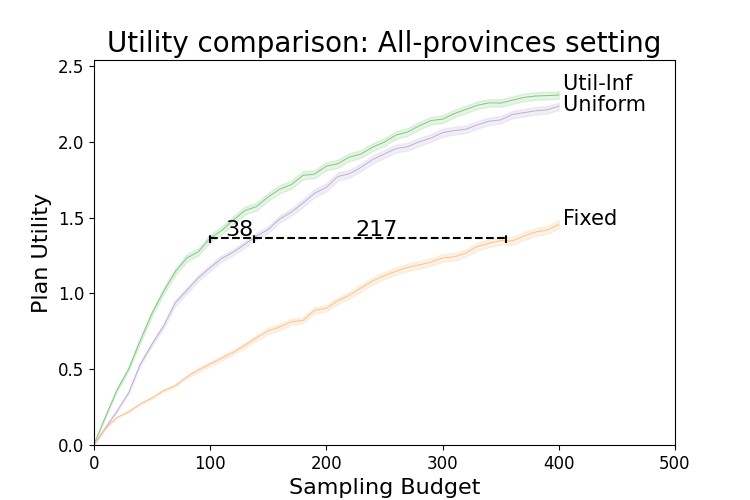}}
    \caption[Utility-informed plans and a utility comparison with the uniform and fixed approaches]{{Utility-informed plans and a utility comparison with the uniform and fixed approaches in the existing and all-provinces settings. 
    The utility comparisons show the number of samples saved under the utility-informed plan, at 100 samples, to attain the same utility under other approaches.
    Variance in utility estimates stems from the use of different MCMC sets of SFP rates at each 10-test interval of the budget; see Appendix \ref{SuppMat:calcUtil} for details.
    }}    
    \label{fig:CSplots}
\end{figure}
%%% END CASE STUDY FIGURE %%%

We first consider the setting of existing provinces.
The SFP risk assessment classified two provinces as Moderate SFP risk and two provinces as Moderately High SFP risk.
We label provinces by these SFP risk classifications as well as the number of associated samples in the first PMS iteration.
For example, province ``Mod (39)'' is classified as Moderate risk and had 39 first-iteration samples. 
Figure \ref{subfig:CS_fam_alloc} shows the utility-informed allocation.
Mod (17) comprises around half the allocation for any considered budget: a budget of 250 samples results in allocating 200 samples to Mod (17), which is more than twice the maximum of 95 samples from the first iteration.
Conversely, Mod (39) only receives 10 samples for any budget exceeding 40 samples.
These two test nodes have the same risk and only differ by 19 first-iteration samples, yet samples at Mod (17) provide substantially stronger utility.
Allocations for budgets exceeding 250 samples emphasize Mod High (26), with a budget of 380 samples resulting in allocating 150 samples to this province.
Mod High (95) never receives more than 20 samples for any budget.
Thus, neither SFP risk nor prior samples are an exclusive driver of locations' utility.

Figure \ref{subfig:CS_fam_utilcompare} shows the estimated cumulative utility across budgets for the utility-informed, uniform, and fixed policies.
We observe a consistent utility improvement under the utility-informed policy relative to other policies. 
A budget of 100 samples deployed under the utility-informed policy requires about 33 fewer samples than the uniform or fixed policy to attain the same utility.
This level of savings (a third of the budget) can be significant in low-resource settings where sample procurement and testing are frequent barriers to more consistent PMS iterations.
We also observe different policy preferences, depending on the budget, between the uniform and fixed policies.
The fixed policy emphasizes allocations to Mod (17) similarly to the utility-informed policy, and has a higher estimated utility than the uniform policy until around 100 samples.
The uniform policy has higher estimated utility than the fixed policy for budgets exceeding 150 samples.
This switch in policy performance makes sense when considering that the uniform policy more closely aligns with the utility-informed policy as the associated sample budget increases.

We next consider the results of the all-provinces setting.
The dotted lines in the allocation of Figure \ref{subfig:CS_expl_1} correspond to test nodes not visited in the existing setting.
These ``new'' test nodes are named by their assessed SFP risk: there are two test nodes each with Moderate and Moderately High risk.
The all-provinces results employ a sourcing matrix for new nodes formed by a single bootstrap sample of the traces observed in the first iteration.
The utility-informed allocations in the existing and all-provinces settings closely align.
The allocation favors new nodes for budgets below 200 samples, allocating at least 30 samples to each new node.
This allocation thus parallels the uniform allocation, and the corresponding utilities in Figure \ref{subfig:CS_expl_2} are closer than in the existing setting.
The allocation emphasizes the nodes of the existing setting for budgets exceeding 200 samples: Mod (17) receives about twice as many samples as Mod High (26), while \hbox{Mod High (95)} and Mod (39) receive only 20 and 10 samples, respectively.
This emphasis mirrors the utility-informed allocation of the existing setting.
From Figure \ref{subfig:CS_expl_2} we see that the fixed policy faces a comparative disadvantage in not sampling new test nodes: more than 250 fixed samples are required to match the utility of 100 samples of the utility-informed allocation.
At the same time, the fixed policy does not incur any sunk costs associated with visiting untested regions.
The all-provinces case thus illustrates how utility can form a basis for exploring trade-offs in operational costs.

The results of this case study illustrate the value of sampling plan utility.
Considering only assessed SFP risk, or where tests have been previously conducted, may not best produce inference that furthers regulatory decision-making.
Applying plan utility leverages supply-chain information and more efficiently centers plans around regulatory decisions.

% DISCUSSION
%%% NEW SECTION %%%
\section{Conclusion} \label{sec:discussion}
Regulation of SFPs is most urgent in low-resource settings, where the burden of poor-quality products is largest \parencite{ozawa2018}. 
PMS is a crucial activity for understanding the quality of products reaching the public \parencite{who2017-2}. 
Our work provides a means of characterizing and estimating the value of PMS through a sampling plan's utility, and demonstrates how this utility can be employed for the comparison of sampling plans.
Approaches for sampling plan development depend substantially on the availability of a efficiently calculated utility metric that appropriately characterizes the value of a sampling plan. 
Ongoing work at \if0\blind United States Pharmacopeia, under the USAID-funded PQM+ program, \fi \if1\blind[\textit{organization removed for blind review}] \fi 
will integrate sampling plan utility into a tool for field users in low- and middle-income countries. 
These countries currently use the risk-based assessment procedure of \citet{nkansah2018} for developing sampling plans.
Identifying plan utility as proposed here only requires two or three additional parameter assessments by regulators to characterize objectives.
This utility is an improvement in sampling plan development that integrates these assessments with prior tests and supply-chain information.

Our proposed loss for identifying sampling plan utility meets practical needs of PMS.
Using a summation of penalties across nodes simplifies calculation, interpretation, and synthesis with data sources that are often available on a node-by-node basis.
We multiple the score by the weight to address the fact that deviations are more important near crucial decision points.
Another benefit of this summation of penalties across nodes is the ability to distinguish loss among nodes.
For instance, catchment areas of test nodes, or market shares of supply nodes, can be used as prioritization terms for each node that can readily incorporated at each node-wise addend.
Evaluation of these terms is often available as part of the assessment process of \citet{nkansah2018}.
However, this summation ignores the interdependence of SFP rates, illustrated by the likelihood in Section \ref{subsec:inferenceinPMS}.

\section*{Data availability}
Due to the nature of this research, participants of this study did not agree for their data to be shared publicly; supporting data are not available. 
In lieu of data used in the case study, analogous data comprising the example of Section \ref{subsec:designUtilExample} that similarly demonstrate the application of our efficient estimation method can be found on Github at 
\if1\blind[\textit{link removed for blind review}]\fi
\if0\blind \url{https://github.com/eugenewickett/measuringUtility_paper_repo}\fi.

\section*{Funding}
\if0\blind
This work was funded through two National Science Foundation grants: EAGER Award 1842369: ISN: Unraveling Illicit Supply Chains for Falsified Pharmaceuticals with a Citizen Science Approach, and NSF 1953111.
We acknowledge United States Pharmacopeial Convention (USP) and the Promoting the Quality 
of Medicines Plus (PQM+) program, funded by the United States Agency for International Development 
(USAID), for the use of data in our case study.
\fi
\if1\blind
[\textit{Funding sources removed for blind review}].
\fi

%%% BIBLIOGRAPHY %%%

\patchbibmacro{date+extradate}{%
  \printtext[parens]%
}{%
  \setunit{\addperiod\space}%
  \printtext%
}{}{}

\singlespacing % For references and appendices

\printbibliography

%TC:ignore

\clearpage
\begin{appendices}

% \newgeometry{top=1in,bottom=1in}
\section{Table of notation} \label{SuppMat:notationTable}
% NOTATION TABLE

\newcolumntype{j}{X}
\newcolumntype{s}{>{\hsize=.3\hsize}X}
\newcommand{\secspace}{0.3cm}
\newcommand{\secspacetwo}{0.1cm}

\begin{table}[h!]
\footnotesize
\caption{ { Notation used in this paper.} } 

\begin{tabularx}{\textwidth}{@{} s j @{}}

\toprule
    
    \textbf{Sets} \\[\secspacetwo]
    $\tnSet$ & Set of test nodes \\
    $\snSet$ & Set of supply nodes \\[\secspace]
    
    \textbf{Decision variables} \\[\secspacetwo]
    $\allocVec \in (\mathbb{Z}^{+})^{|\tnSet|}$ & Vector denoting the sampling plan, where $\allocVec_\tnPoint$ is the number of tests allocated to $\tnPoint\in\tnSet$ \\[\secspace]
  
    \textbf{Given parameters} \\[\secspacetwo]
    $\numTests\in\posIntSet$ & Number of prior tests \\
    $\numTestsNext\in\posIntSet$ & Budget, in tests, for the next PMS iteration \\
    % $\sourcingMat\in[0,1]^{|\tnSet|\times|\snSet|}$ & Matrix of sourcing probabilities, where $\sourcingMat_{\tnPoint\snPoint}$ is the probability test node $\tnPoint\in\tnSet$ sources from supply node $\snPoint\in\snSet$ \\
    $\bm{\testResultpoint}\in\{0,1\}^{\numTestsNext}$ & Vector of binary test results, where $\testResultpoint_i$ is the result of the $i^{\text{th}}$ test \\
    $\bm{\tnPoint}\in\tnSet^{\numTestsNext}$ & Vector of test-node labels, where $\bm\tnPoint_i$ is the label associated with the $i^{\text{th}}$ test \\
    $\bm{\snPoint}\in\snSet^{\numTestsNext}$ & Vector of supply-node labels, where $\bm\snPoint_i$ is the label associated with the $i^{\text{th}}$ test \\
    $\numTestsMat\in(\posIntSet)^{|\tnSet|\times|\snSet|}$ & Matrix of prior tests along each trace, where $\numTestsMat_{\tnPoint\snPoint}$ is the number of tests traversing supply node $\snPoint$ and test node $\tnPoint$ \\
    $\testResultMat\in(\posIntSet)^{|\tnSet|\times|\snSet|}$ & Matrix of prior SFPs along each trace, where $\testResultMat_{\tnPoint\snPoint}$ is the number of observed SFPs traversing supply node $\snPoint$ and test node $\tnPoint$ \\
    $\dataSet$ & Prior PMS data set \\
    $\dataSetNext$ & Prospective PMS data set \\
    $\tnRates\in(0,1)^{|\tnSet|}$ & Vector of test-node SFP rates, where $\tnRates_\tnPoint$ is the SFP rate at test node $\tnPoint$; can likewise be defined on the space of reals, $\realSet^{|\tnSet|}$, through a logit transformation \\
    $\snRates\in(0,1)^{|\snSet|}$ & Vector of supply-node SFP rates, where $\snRates_\snPoint$ is the SFP rate at supply node $\snPoint$; can equivalently be defined on the space of reals, $\realSet^{|\snSet|}$ \\
    $\tnsnRates\in(0,1)^{|\tnsnSet|}$ & Vector of SFP rates at test nodes and supply nodes; can equivalently be defined on the space of reals, $\realSet^{|\tnsnSet|}$ \\[\secspace]
     
    \textbf{User-defined \hspace{20pt} parameters} \\[\secspacetwo]
    % $\priortiz(\tnsnPoint)$  &  Prioritization for accurate estimation of the SFP rate at node $\tnsnPoint$ \\
    $\threshold\in(0,1)$ & Classification threshold for significantly high SFP rates  \\
    $\underEstWt\in\realSet^{+}$ & Underestimation penalty relative to overestimation \\
    $\checkRiskSlope\in[0,1]$ & Tuning parameter for the assessment weight \\[\secspace]

    \textbf{Functions} \\[\secspacetwo]
    $\consolSFPfunc_{\tnPoint\snPoint}(\tnRates,\snRates)$ & Consolidated SFP rate along trace $(\tnPoint,\snPoint)$ for SFP rates $(\tnRates,\snRates)$, or the probability of being SFP for a sample traversing $\tnPoint$ and $\snPoint$ \\
    $\estVal(\tnsnRatesEst,\tnsnRates)$   & Loss for SFP-rate estimate $\tnsnRatesEst$ under latent SFP rates $\tnsnRates$ \\
    $\score(\tnsnRatesEst_\tnsnPoint,\tnsnRates_\tnsnPoint)$   & Score for SFP-rate estimate $\tnsnRatesEst_\tnsnPoint$ under latent SFP rate $\tnsnRates_\tnsnPoint$ \\
    $\weight(\tnsnRates_\tnsnPoint)$  &  Weight for the estimation importance of latent SFP rates $\tnsnRates_\tnsnPoint$ \\
    % $\consolInaccFunc_{\tnPoint\snPoint}(\tnRates,\snRates)$ & SFP classification probability along trace $(\tnPoint,\snPoint)$ for SFP rates $(\tnRates,\snRates)$ \\
    $\classFunc(\tnsnRates_\tnsnPoint)$ & Classification of node $\tnsnPoint$ as either a significant or insignificant source of SFPs, as a function of threshold $\threshold$ and SFP rate $\tnsnRates_\tnsnPoint$  \\
    $\util (\allocVec)$ & Utility, or expected reduction in loss, for sampling plan $\allocVec$ \\
    
\bottomrule
\end{tabularx}

\end{table}

\normalsize
% \restoregeometry
\clearpage
%%% NEW SECTION %%%
\section{Estimating sampling plan utility} \label{SuppMat:calcUtil}
Calculating utility using standard Markov Chain Monte Carlo (MCMC) integration is computationally prohibitive for many low-resource contexts. 
Our efficient method leverages the binary nature of PMS data to produce estimates of plan utility that are computed considerably faster than the utility estimates estimated through MCMC integration.
Synthesis with importance sampling addresses estimate bias in settings with many nodes and many tests.

%%% NEW SUBSECTION %%%
\subsection{MCMC integration for plan utility} \label{subsec:approxDesignUtil}
The utility $U(\allocVec)$ defined in Section \ref{sec:samplingDesigns} is the difference between $\expecOp\estVal\big[(\tnsnRatesEst(\allocVec),\tnsnRates\big]$ and $\expecOp\estVal\big[(\tnsnRatesEst(\bm{0}),\tnsnRates\big]$, i.e., the utility is the difference in expected loss obtained when deploying sampling plan $\allocVec$ versus conducting no additional tests.
Using the Bayesian experimental construction of \citet{chaloner1995}, the expected loss can be expressed as
\begin{equation} \label{eq:mathFormOfLoss}
    \expecOp\estVal\big[(\tnsnRatesEst(\allocVec),\tnsnRates\big] = 
    \int_{\tnsnRates \in \realSet^{|\tnSet|+|\snSet|}}
    \inf_{\Tilde\tnsnRates(\cdot)}
    \int_{\dataSetNext\in{\mathcal{D}(\allocVec)}} \estVal\big[\Tilde\tnsnRates(\dataSet,\dataSetNext),\tnsnRates\big] p(\tnsnRates|\dataSet,\dataSetNext) p(\dataSetNext|\dataSet) 
    \hspace{4pt} \mathrm{d}\dataSetNext \mathrm{d}\tnsnRates\,,
\end{equation}
where $\mathcal{D}(\allocVec)$ is the set of all possible data sets resulting from plan $\allocVec$, and $\dataSet$ is the set of existing data.
Each member of $\mathcal{D}(\allocVec)$ is a set of $\numTestsNext$ test results with associated test node-supply node traces.
For each set of SFP rates, $\tnsnRates$, and new data set, $\dataSetNext\in\mathcal{D}(\allocVec)$, the associated loss $\estVal\big[\Tilde\tnsnRates(\dataSet,\dataSetNext),\tnsnRates\big]$ is weighted by the likelihood of the new data set given the existing data, $p(\dataSetNext|\dataSet)$, as well as the likelihood of the SFP rates given the new data and existing data, $p(\tnsnRates|\dataSet,\dataSetNext)$.
The infimum captures the use of a Bayes estimate, the loss-minimizing SFP-rate estimate under $\dataSetNext$.
Calculation of $\expecOp\estVal\big[(\tnsnRatesEst(\allocVec),\tnsnRates\big]$ requires integration over possible $\dataSetNext$ and $\tnsnRates$.
This integration requires posteriors $p(\tnsnRates|\dataSet)$ and $p(\tnsnRates|\dataSet,\dataSetNext)$. 
These posteriors can be challenging to calculate directly: \citet{kleingesse2019} noted that Monte Carlo integration is standard for estimating expressions like $\expecOp\estVal\big[(\tnsnRatesEst(\allocVec),\tnsnRates\big]$.
The posterior of \citet{wickett2023} is also computed through MCMC sampling.

%%% MCMC INTEGRATION ALGORITHM %%%
\def\lc{\left\lceil}   
\def\rc{\right\rceil}
\begin{algorithm}[!ht]
    \DontPrintSemicolon
        \KwInput{$\allocVec$, $\tnSet$, $\snSet$, $\sourcingMat$, $\dataSet$, $\estVal(\cdot,\cdot)$, $\algparamone$, $\algparamtwo$}
        \KwOutput{Estimate of $\expecOp\big[\estVal(\tnsnRatesEst(\allocVec),\tnsnRates)\big]$}
        \SetKwProg{Fn}{Function}{:}{}
        Set $\truthdraws:=$ set of $\algparamone$ MCMC draws from $p(\tnsnRates|\dataSet)$. \\
        \For{each $j$ in $\{1,\dots,\algparamtwo\}$}{
        Sample $\tnsnRates^{(j)}$ from $\truthdraws$. \\
            \For{each $\tnPoint$ in $\tnSet$}{
                \For{each $i$ in $\{1,\dots,\allocVec_{a}\}$ }{
                    Set $\bm{a}_i:=a$, and sample $\bm{b}_i$ according to $\sourcingMat_{\tnPoint}$.\\
                    Sample $\bm{\testResultpoint}_i$ from Bern$\big[\consolInaccFunc_{\bm{a}_i\bm{b}_i}(\tnsnRates^{(j)})\big]$.
                }             
            }
            Set $\dataSetNext:=(\testResultVecNext, \bm{\tnPoint}, \bm{\snPoint})$. \\
            Generate $\mcmcdatadraws_j$, a set of $\algparamone$ MCMC draws from $p(\tnsnRates|\dataSet,\dataSetNext)$. \\
            Set $\tnsnRatesEst:=\argmin\limits_{\Tilde\tnsnRates} \sum\limits_{\tnsnRates\in\mcmcdatadraws_j}\estVal(\Tilde\tnsnRates,\tnsnRates)\,.$ \\
            Set $\bm{L}_{j}:=\frac{1}{|{\mcmcdatadraws_j}|}\sum\limits_{\tnsnRates\in\mcmcdatadraws_j}\estVal(\tnsnRatesEst, \tnsnRates)$.
        }
        Return $\frac{1}{\algparamtwo}\sum\limits_{j\in\{1,\dots,\algparamtwo
        \}}\bm{\estVal}_{j}$.
    \caption{MCMC estimation of $\expecOp[\estVal(\tnsnRatesEst(\allocVec),\tnsnRates)]$}
    \label{alg:mcmcIntegration}
\end{algorithm}
%%% END MCMC INTEGRATION ALGORITHM %%%

Our Monte Carlo integration procedure for $\expecOp\big[\estVal(\tnsnRatesEst(\allocVec),\tnsnRates)\big]$ is given in Algorithm \ref{alg:mcmcIntegration}.
The expected loss is approximated through three key sets: an MCMC set of SFP-rate vectors drawn from the posterior for the SFP rates under existing data; a set of new data sets simulated using these posterior draws; and an MCMC set of SFP-rate vectors drawn from the posterior for the SFP rates under existing and new data.
An SFP-rate draw from the posterior under existing data, $p(\tnsnRates|\dataSet)$, is taken as the latent truth.
New data, $\dataSetNext$, are simulated under this latent truth as well as the sampling plan, $\allocVec$.
Draws from a posterior under these new data and existing data, $p(\tnsnRates|\dataSet,\dataSetNext)$, are used to identify a Bayes estimate and the loss associated with the Bayes estimate.
Averaging the loss across latent SFP-rate draws approximates the expected loss.
The key algorithm parameters are the size of the sets of SFP-rates, $\algparamone$, and the number of simulated new data sets, $\algparamtwo$.
Low values for $\algparamone$ induce bias: Bayes estimates are identified using an insufficient representation of $p(\tnsnRates|\dataSet)$, such that the true expected loss is underestimated.
Low values for $\algparamtwo$ induce variance: the range of possible data sets is insufficiently explored.
Accordingly, both larger PMS budgets and settings with lower levels of initial data require a higher $\algparamtwo$ to achieve a desired variance: the range of possible data sets is larger.

% \vspace{-10pt}
Algorithm \ref{alg:mcmcIntegration} uses as inputs the sampling plan allocation vector ($\allocVec$), the sets of test and supply nodes $(\tnSet,\snSet)$, the sourcing-probability matrix $(\sourcingMat)$, existing data $(\dataSet)$, a specified loss $(\estVal)$, the size of the MCMC sets of SFP rates $(\algparamone)$, and a number of data sets to generate $(\algparamtwo)$.
Each of $\algparamtwo$ data sets is a single simulation $\dataSetNext$ of the data obtained under plan $\allocVec$.
The $\algparamone$ MCMC draws from $p(\tnsnRates|\dataSet)$ or $p(\tnsnRates|\dataSet,\dataSetNext)$ enable numerical integration.
% Algorithm \ref{alg:mcmcIntegration} is completed in four stages, as follows.
% \paragraph{Generate SFP rates from existing data} 
% In Line 1, $\algparamone$ SFP-rate vectors are drawn under existing data $\dataSet$.
% These vectors are stored as $\truthdraws$, where draw $\tnsnRates^{(j)}\in\truthdraws$ is a set of SFP rates for all nodes in $\tnSet$ and $\snSet$.
% These draws are used to simulate new data sets.
% \paragraph{Simulate a new data set} 
% In Line 3, vector of SFP rates $\tnsnRates^{(j)}$ is the current latent truth. 
% In Lines 4 to 8, new data set $\dataSetNext$ is generated.

% \vspace{5pt}
New data generation uses $\tnsnRates^{(j)}$, a vector of SFP rates capturing the latent truth, and sampling plan $\allocVec$, where $\allocVec_\tnPoint$ is the number of tests allocated to test node $\tnPoint\in\tnSet$. 
Products at $\tnPoint$ are sourced from supply nodes in $\snSet$ according to sourcing-probability matrix $\sourcingMat$. 
Element $\sourcingMat_{\tnPoint\snPoint}\in[0,1]$ is the probability that a product from $\tnPoint$ is sourced from supply node $\snPoint\in\snSet$. $\sum_{\snPoint\in\snSet}\sourcingMat_{\tnPoint\snPoint}=1$ for each $\tnPoint\in\tnSet$.
The supply node associated with each test $i$ at test node $\tnPoint$ is drawn according to the sourcing probabilities of vector $\sourcingMat_\tnPoint$. 
A Bernoulli test result $\bm{y}_i$ is generated for test $i$ from the consolidated SFP probability of (\ref{eq:consolSFPProb}); the consolidated SFP probability is the SFP likelihood when accounting for test-node and supply-node SFP rates.
% The consolidated SFP probability uses the drawn supply node, the corresponding test node, and the latent SFP-rate vector.
% New data set $\dataSetNext$ has a test result, a test-node label, and a supply-node label for each test indicated by the sampling plan.

% \paragraph{Use new data to estimate simulation loss}
% The next step is determining the estimate decision and associated loss under the new data.
% This determination requires an updated set of SFP-rate draws.
% In Line 9, new data $\dataSetNext$ is used to generate $\Gamma$, a set of MCMC draws from $p(\tnsnRates|\dataSet,\dataSetNext)$.
% In Line 10, Bayes estimate $\tnsnRatesEst$ is selected as the member of $\Gamma$ that minimizes the sum loss across all members of $\Gamma$.
% Line 11 uses the average loss under $\tnsnRatesEst$ as the loss estimate for simulation $j$.

% \paragraph{Create estimate of expected loss} 
The average loss across all $\algparamtwo$ simulations yields the estimate of $\expecOp\big[\estVal(\tnsnRatesEst(\allocVec),\tnsnRates)\big]$.
An approximate $(1-\alpha)\%$ confidence interval for this average loss can be formed via the standard Central Limit Theorem formulation
\begin{equation*}
   \frac{1}{\algparamtwo}\sum_{j\in\{1,\dots,\algparamtwo\}}\bm{\estVal}_{j} \pm z_{1-\frac{\alpha}{2}}\frac{\hat{\sigma}}{\sqrt{\algparamtwo}}  \,,
\end{equation*}
where $z_{1-\frac{\alpha}{2}}$ is the standard normal value at the $(1-\frac{\alpha}{2})^{\text{th}}$ quantile and $\hat\sigma$ is the standard deviation of observed simulation utility realizations.
This confidence interval can be similarly applied under the efficient and importance sampling algorithms in Sections \ref{subsec:fastalg} and \ref{subsec:importancesampling}, and is displayed in the figures of the case study of Section \ref{sec:casestudy}.

Algorithm \ref{alg:mcmcIntegration} can consume substantial computational resources. 
$\algparamone$ MCMC samples are needed for each of $\algparamtwo$ simulations, and each round of MCMC sampling takes a few minutes to conduct with an Intel Core 1.6GHz processor.
Using $\algparamone=10,000$ and $\algparamtwo=1,000$ requires about 16 hours to complete an expected loss estimate for one plan designation of $\allocVec$ in the six-node example of Section \ref{subsec:designUtilExample}.
% For regulators with limited computational resources, MCMC integration cannot enable usable evaluations of considered sampling plans, nor easily allow analysis of the marginal impact of budgets.
% A faster algorithm is needed for applications of plan utility to be practical.

%%% NEW SUBSECTION %%%
\subsection{A faster algorithm for plan utility} \label{subsec:fastalg}

A key factor in the long processing time of Algorithm \ref{alg:mcmcIntegration} is the need for repeated MCMC sampling of the posterior after the generation of each new data set.
Reformulating the expression $p(\tnsnRates|\dataSet,\dataSetNext) p(\dataSetNext|\dataSet)$ of (\ref{eq:mathFormOfLoss}) yields a useful equivalence:
\begin{equation*}
\resizebox{.99\hsize}{!}{$
     p(\tnsnRates|\dataSet,\dataSetNext) p(\dataSetNext|\dataSet) = \frac{ p(\tnsnRates,\dataSet,\dataSetNext) p(\dataSet,\dataSetNext)}{p(\dataSet,\dataSetNext)p(\dataSet)}  = 
     % \frac{p(\tnsnRates,\dataSet,\dataSetNext)}{p(\dataSet)} = 
     \frac{p(\dataSetNext|\dataSet,\tnsnRates) p(\dataSet,\tnsnRates)}{p(\dataSet)} = p(\tnsnRates|\dataSet) p(\dataSetNext|\dataSet,\tnsnRates) \,.
$}
\end{equation*}
Conditional on $\tnsnRates$, $\dataSet$ and $\dataSetNext$ are independent; thus,
\begin{equation} \label{eq:postRewrite}
p(\tnsnRates|\dataSet,\dataSetNext) p(\dataSetNext|\dataSet) = p(\tnsnRates|\dataSet) p(\dataSetNext|\tnsnRates) \,.    
\end{equation}
The first term on the right-hand side, $p(\tnsnRates|\dataSet)$, is the posterior for SFP rates $\tnsnRates$ under existing data $\dataSet$.
The second term, $p(\dataSetNext|\tnsnRates)$, is the binomial likelihood of a data set for a given set of SFP rates:
\begin{equation} \label{eq:binomlike}
    \prod\limits_{\tnPoint\in\tnSet} \prod\limits_{\snPoint \in\snSet}\consolInaccFunc_{\tnPoint\snPoint}(\tnsnRates)^{{\testResultNextMat}_{\tnPoint\snPoint}}  \bigl[1-\consolInaccFunc_{\tnPoint\snPoint}(\tnsnRates) \bigr]^{{\numTestsNextMat}_{\tnPoint\snPoint}-{{\testResultNextMat}_{\tnPoint\snPoint}}} \genfrac(){0pt}{}{{\numTestsNextMat}_{\tnPoint\snPoint}}{{\testResultNextMat}_{\tnPoint\snPoint}}\,, 
\end{equation}
where $\numTestsNextMat_{\tnPoint\snPoint}$ and $\testResultNextMat_{\tnPoint\snPoint}$ are the respective numbers of tests and SFP positives in $\dataSetNext$ along trace $(\tnPoint,\snPoint)$. 

The equivalence of (\ref{eq:postRewrite}) feeds an efficient approximation for the expected loss.
Let $\Gamma\sim p(\tnsnRates|\dataSet)$ be a set of MCMC draws under $\dataSet$, let $\tnsnRates^{(i)}\in\Gamma$ be the $i^{\text{th}}$ member of $\Gamma$, let $\mathcal{D}(\allocVec)$ be the set of possible data sets generated by plan $\allocVec$, and let $\mathcal{D}_{\Gamma}(\allocVec)$ be a set of $|\Gamma|$ data sets randomly generated under plan $\allocVec$ for each SFP-rate vector in $\Gamma$.
Substituting the right-hand side of (\ref{eq:postRewrite}) into (\ref{eq:mathFormOfLoss}) yields an approximation of the expected loss that can be calculated without repeated MCMC sampling:

\begin{align} \label{eq:approxFormOfLoss}
    \mathbbm{E}\estVal\big[(\tnsnRatesEst(\allocVec),\tnsnRates\big] &= 
    \int_{\dataSetNext\in{\mathcal{D}(\allocVec)}}
    \inf_{\tnsnRatesEst(\cdot)}
    \int_{\tnsnRates \in \realSet^{|\tnSet|+|\snSet|}} 
     \estVal\big[\tnsnRatesEst(\dataSet,\dataSetNext),\tnsnRates\big] p(\tnsnRates|\dataSet) p(\dataSetNext|\tnsnRates) 
    \hspace{4pt} \mathrm{d}\tnsnRates \mathrm{d}\dataSetNext \\
    &\approx \int_{\dataSetNext\in{\mathcal{D}(\allocVec)}} \inf_{\tnsnRatesEst(\cdot)} \label{eq:step1}
    \frac{1}{|\Gamma|}
    \sum\limits_{i\in\{1,\dots,|\Gamma|\}} 
    \estVal\big[\tnsnRatesEst(\dataSet,\dataSetNext),\tnsnRates^{(i)}\big]  p(\dataSetNext|\tnsnRates^{(i)})
    \hspace{4pt} \mathrm{d}\dataSetNext \\
    &\leq \int_{\dataSetNext\in{\mathcal{D}(\allocVec)}}  \min_{\Tilde\tnsnRates:\mathcal{D}(\allocVec)\rightarrow \realSet^{|\tnSet|+|\snSet|} } \label{eq:step2}
    \frac{1}{|\Gamma|}
    \sum\limits_{i\in\{1,\dots,|\Gamma|\}} 
    \estVal\left[\Tilde\tnsnRates(\dataSetNext),\tnsnRates^{(i)}\right]  p(\dataSetNext|\tnsnRates^{(i)})
    \hspace{4pt} \mathrm{d}\dataSetNext \\
    &\approx \sum\limits_{\dataSetNext^{(j)}\in\mathcal{D}_{\Gamma}(\allocVec)}
    \min_{\Tilde\tnsnRates:\mathcal{D}(\allocVec)\rightarrow \realSet^{|\tnSet|+|\snSet|}} \label{eq:step3}
    \frac{1}{|\Gamma|}
    \sum\limits_{i\in\{1,\dots,|\Gamma|\}}
    \estVal\left[\Tilde\tnsnRates(\dataSetNext^{(j)}),\tnsnRates^{(i)}\right]  p(\dataSetNext^{(j)}|\tnsnRates^{(i)}) \,.
\end{align} 
The approximation in (\ref{eq:step1}) proceeds from substituting the initial set of posterior MCMC draws of SFP rates.
The inequality in (\ref{eq:step2}) results from using MCMC draws of SFP rates to identify the Bayes estimate: these draws are a finite set and do not fully capture the posterior.
The second approximation in (\ref{eq:step3}) follows from using sets of generated data instead of integrating over all possible data sets.
The final expression can be produced with one set of posterior MCMC draws of SFP rates and computationally inexpensive binomial likelihoods. 

%%% EFFICIENT ALGORITHM FOR EXPECTED LOSS %%%
\def\lc{\left\lceil}   
\def\rc{\right\rceil}
\begin{algorithm}[!ht]
    \DontPrintSemicolon
        \KwInput{$\allocVec$, $\tnSet$, $\snSet$, $\sourcingMat$, $\dataSet$, $\estVal(\cdot,\cdot)$, $\algparamone$, $\algparamtwo$}
        \KwOutput{Estimate of $\expecOp\big[\estVal(\tnsnRatesEst(\allocVec),\tnsnRates)\big]$}
        \SetKwFunction{FMain}{MakeEstimate}
        \SetKwFunction{FSub}{VecTNSamples}
        \SetKwFunction{FSubc}{PossibleNSets}
        \SetKwFunction{FSubd}{PossibleData}
        \SetKwFunction{FbuildW}{BuildDatasimMatrix}
        \SetKwProg{Fn}{Function}{:}{}
        \Fn{\FbuildW{$\truthdraws$,  $\datadraws$, $\allocVec$, $\tnSet$, $\snSet$, $\sourcingMat$, $\algparamone$, $\algparamtwo$}}{
            Initialize $\binomLikeMat\in\realSet^{\algparamone\times \algparamtwo}$.\\
            \For{$j$ in $\{1,\dots,\algparamtwo\}$}{
                Generate data $\dataSetNext^{(j)}:=(\bm{y},\bm{a},\bm{b})$ according to $\allocVec$, $\sourcingMat$ and $\tnsnRates^{(j)}\in\datadraws$. \\
                \For{$\tnPoint\in\tnSet$, $\snPoint\in\snSet$}{
                    Set $\numTestsNextMat_{\tnPoint\snPoint}:=\sum\limits_{i\in\{1,\dots,\numTestsNext\}}\mathbbm{1}\{\bm{a}_i=a,\bm{b}_i=b\}$.\\
                    Set $\testResultNextMat_{\tnPoint\snPoint}:=\sum\limits_{i\in\{1,\dots,\numTestsNext\}}\mathbbm{1}\{\bm{a}_i=a,\bm{b}_i=b,\bm{y}_i=1\}$. 
                }
                \For{$\tnsnRates^{(i)}\in\truthdraws$}{
                    Set $\binomLikeMat_{ij}:=\prod\limits_{\tnPoint\in\tnSet}\prod\limits_{\snPoint\in\snSet}\consolInaccFunc_{\tnPoint\snPoint}\big(\tnsnRates^{(i)}\big)^{{ \testResultNextMat}_{\tnPoint\snPoint}}  \big[1-\consolInaccFunc_{\tnPoint\snPoint}\big(\tnsnRates^{(i)}\big)\big]^{{\numTestsNextMat}_{\tnPoint\snPoint}-{{\testResultNextMat}_{\tnPoint\snPoint}}}\genfrac(){0pt}{}{{\numTestsNextMat}_{\tnPoint\snPoint}}{{\testResultNextMat}_{\tnPoint\snPoint}}$.
                }
                Normalize column $\binomLikeMat_{\cdot j}$ to sum to 1.
            }
            Return $\binomLikeMat$.
        }
        Set $\truthdraws, \datadraws:=\algparamone,\algparamtwo$ MCMC draws from $p(\tnsnRates|\dataSet)$. \\
        Get $\binomLikeMat:=$\texttt{BuildDatasimMatrix}($\truthdraws$,  $\datadraws$, $\allocVec$, $\tnSet$, $\snSet$, $\sourcingMat$, $\algparamone$, $\algparamtwo$).\\
        \For{$j$ in $\{1,\dots,\algparamtwo\}$}{
            Set  $\bm{u}_j = \min_{\tnsnRatesEst} \bigg\{ \sum_{i=1}^{\algparamone}\binomLikeMat_{ij} L\big(\tnsnRatesEst, \tnsnRates^{(i)}\big) \bigg\}$
        }
        Return $\frac{1}{\algparamtwo}\sum_{j=1}^{\algparamtwo} \bm{u}_{j}$.
    \caption{Efficient estimation of $\expecOp\big[\estVal\big(\tnsnRatesEst(\allocVec),\tnsnRates\big)\big]$}
    \label{alg:utilApproxNodeFast}
\end{algorithm}
%%% END EFFICIENT ALGORITHM %%%

Algorithm \ref{alg:utilApproxNodeFast} leverages the approximation of (\ref{eq:step3}) to produce an efficient estimate of the expected loss.
% A flow chart illustrating this algorithm is provided in Appendix \ref{SuppMat:flowChart}.
This algorithm uses the same inputs as Algorithm \ref{alg:mcmcIntegration}.
A single MCMC set of SFP-rate vectors and a set of of simulated data sets are used.
The equivalence of (\ref{eq:postRewrite}) means that incorporating the likelihood of each SFP rate with respect to each simulated data set is sufficient for approximating the expected loss: no further MCMC sampling is needed.
A key matrix for this algorithm, $\binomLikeMat$, stores these likelihoods: each row corresponds to an SFP-rate vector and each column corresponds to a simulated data set.
A Bayes estimate is derived for each simulated data set through a weighting of the loss corresponding to each SFP-rate vector by the binomial likelihoods in $\binomLikeMat$.
The set of weighted losses for each simulated data set are averaged to yield the expected loss.
% The algorithm details are as follows.
% \paragraph{Build $\binomLikeMat$}
% Data simulation matrix $\binomLikeMat$ is constructed in Lines 1 to 13.
% Parameters $\algparamone$ and $\algparamtwo$ determine the respective sizes of sets $\truthdraws$, the set of MCMC draws from $p(\tnsnRates|\dataSet)$, and $\datadraws$, the set of draws producing possible data sets.
% As in Algorithm \ref{alg:mcmcIntegration}, each data set $\dataSetNext^{(j)}$ is simulated according to the plan's allocation ($\allocVec$), the sourcing-probability matrix ($\sourcingMat$), and the SFP rates of the simulation ($\tnsnRates^{(j)}$). 
% Element $\binomLikeMat_{ij}$ is the binomial data likelihood for $\tnsnRates^{(i)}\in\truthdraws$ under $\dataSetNext^{(j)}$.
% The principal difference from the MCMC estimation of Algorithm \ref{alg:mcmcIntegration} is that these likelihoods determine the posterior for SFP rates under new data; more MCMC sampling for this posterior is not needed.
% The initial MCMC draws in $\truthdraws$ account for $\dataSet$, and the likelihoods account for the simulated new data.
In Line 10, the data likelihoods are normalized column-wise to sum to one, as  the data are generated from members of $\datadraws$ that are distributed according to $p(\tnsnRates|\dataSet)$. 
Failure to normalize would disproportionately affect data sets corresponding to different members of $\datadraws$.
% \paragraph{Identify minimum loss under each simulated data set} 
In Lines 14 to 15, the minimum loss $\bm{u}_j$ for each simulated data set $\dataSetNext^{(j)}$ is identified.
Bayes estimate $\tnsnRatesEst$ minimizes the loss across the members of $\truthdraws$ approximating the distribution $p(\tnsnRates|\dataSet)$, weighted for the binomial likelihoods in column $\binomLikeMat_{\cdot j}$.
Losses using the scores described in Sections \ref{subsubsec:objClass} and \ref{subsubsec:objAssess} have analytical minimizers, shown in Appendix \ref{SuppMat:bayesEstimates}; thus, Bayes estimates under these scores can be found through an indexing of the weighted losses.   
Our Python implementation uses these analytical solutions or, for other scores, the low-memory Broyden-Fletcher-Goldfarb-Shanno (L-BFGS) minimization algorithm available in SciPy \parencite{scipy2020}. 
% The average of all $\bm{u}_j$ in Line 17 then approximates the expected loss $\mathbbm{E}[\estVal(\tnsnRatesEst(\allocVec),\tnsnRates)]$.
% This collection of $\bm{u}_j$ values can further be used to form standard confidence intervals, as described for MCMC estimation.

The principal parameters of Algorithm \ref{alg:utilApproxNodeFast} are $\algparamone$, the number of MCMC draws for representing $p(\tnsnRates|\dataSet)$, and $\algparamtwo$, the number of data sets to simulate.
Computational experiments show it is beneficial to expend more resources on $\truthdraws$ than $\datadraws$; see
\if0\blind{\citet{wickettdissertation} }\fi
\if1\blind{\citet{wickettdissertationblind} }\fi
for the impact of different $\algparamone$ and $\algparamtwo$ choices on utility estimates.
Low values for $\algparamone$ induce bias: identified minima insufficiently account for a full representation of $p(\tnsnRates|\dataSet)$, such that the true expected loss is underestimated.
This bias will tend to increase with larger PMS budgets, as the Bayes estimate is more likely to move to areas of the parameter space less represented in $\truthdraws$. 
Low values for $\algparamtwo$ induce variance: the range of possible data sets is insufficiently explored.
We suggest choosing $\algparamone$ as large as computational resources will allow, and $\algparamtwo$ as large as needed to meet variance tolerance.

Algorithm  \ref{alg:utilApproxNodeFast} runs significantly faster than the MCMC integration of Algorithm \ref{alg:mcmcIntegration}.
The execution of full MCMC integration needs multiple hours; this efficient implementation calculates around four simulated data sets per second for our case study.
Our case study uses $\algparamone=75,000$ and $\algparamtwo=2,000$, as these values produce consistent allocation decisions
\if0\blind{\parencite{wickettdissertation}. }\fi
\if1\blind{\parencite{wickettdissertationblind}. }\fi
This processing rate means utility estimation takes about 10 minutes on a laptop with an Intel Core 1.6GHz processor.

%%% NEW SUBSECTION %%%
\subsection{Using importance sampling to reduce bias and variance} \label{subsec:importancesampling}

The quality of utility estimates produced through the efficient estimation algorithm is impacted by four key factors: the size of the set of MCMC draws from $p(\tnsnRates|\dataSet)$, $\algparamone$; the size of the set of MCMC draws simulating data sets, $\algparamtwo$; the number of nodes, $|\tnSet|+|\snSet|$; and the number of tests in a prospective sampling plan, $\numTestsNext$
\if0\blind{\parencite{wickettdissertation}. }\fi
\if1\blind{\parencite{wickettdissertationblind}. }\fi
A sufficiently high number of nodes or number of tests means a small subset of SFP-rate vectors in the MCMC set will have a large impact on the resulting loss estimate, i.e., the significant region of SFP rates will be underrepresented.

The significant region of SFP rates can be be better represented through importance sampling.
We form loss estimates through an additional MCMC set of SFP-rate vectors better capturing this significant region, drawn using an ``expected'' set of traces and test results.
These loss estimates are then adjusted to account for the oversampling of the significant region. 
Let $\dataSetImp$ be the expected data set under sampling plan $\allocVec$.
Rewriting the posterior expression in (\ref{eq:postRewrite}) yields the basis for importance sampling:
\begin{equation} \label{eq:postImportance}
    p(\tnsnRates|\dataSet) p(\dataSetNext|\tnsnRates) = p(\tnsnRates|\dataSet, \dataSetImp) \frac{p(\tnsnRates|\dataSet)}{p(\tnsnRates|\dataSet, \dataSetImp)} p(\dataSetNext|\tnsnRates)\,.
\end{equation}
The first term on the right-hand side of ($\ref{eq:postImportance}{}$), $p(\tnsnRates|\dataSet, \dataSetImp)$, is the basis for new MCMC generation of SFP-rate vectors.
The second term, $\frac{p(\tnsnRates|\dataSet)}{p(\tnsnRates|\dataSet, \dataSetImp)}$, indicates scaling of the new SFP-rate vectors into the target space characterized by $\dataSet$.
The third term is the binomial likelihood of a data set, as in (\ref{eq:binomlike}). 

Sampling plan loss is estimated similarly to the efficient method of Algorithm \ref{alg:utilApproxNodeFast} by substituting an MCMC set of SFP-rate vectors generated under $\dataSet$ and $\dataSetImp$ for the MCMC set generated only under $\dataSet$.
The expected data set, $\dataSetImp$, is an assumed data set that allows a better centering of generated SFP-rate vectors at a parameter space important to sampling plan evaluation.
We build $\dataSetImp$ through repeated simulation of the data set resulting from implementation of $\allocVec$.

%%% IMPORTANCE SAMPLING ALGORITHM FOR EXPECTED LOSS %%%
\def\lc{\left\lceil}   
\def\rc{\right\rceil}
\begin{algorithm}[!ht]
    \DontPrintSemicolon
        \KwInput{$\allocVec$, $\tnSet$, $\snSet$, $\sourcingMat$, $\dataSet$, $\estVal(\cdot,\cdot)$, $\algparamone$, $\algparamtwo$, $\algparamthree$}
        \KwOutput{Estimate of $\expecOp\big[\estVal(\tnsnRatesEst(\allocVec),\tnsnRates)\big]$}
        \SetKwFunction{FMain}{MakeEstimate}
        \SetKwFunction{FSub}{VecTNSamples}
        \SetKwFunction{FSubc}{PossibleNSets}
        \SetKwFunction{FSubd}{PossibleData}
        \SetKwFunction{FbuildW}{BuildDatasimMatrix}
        \SetKwFunction{Fbuildimpdata}{BuildImportanceData}
        \SetKwProg{Fn}{Function}{:}{}
        \Fn{\Fbuildimpdata{$\truthdraws$, $\allocVec$, $\sourcingMat$, $\algparamthree$}}{
            Initialize $\mathcal{N}:=\emptyset$,  $\mathcal{Y}:=\emptyset$, $\numTestsNextMat^{\text{imp}}\in(\posIntSet)^{\tnSet\times\snSet}$ and $\testResultNextMat^{\text{imp}}\in(\posIntSet)^{\tnSet\times\snSet}$. \\
            \For{$k\in\{1,\dots,\algparamthree\}$} {
                Simulate $\numTestsNextMat^{(k)}$ according to $\allocVec$ and $\sourcingMat$, where $\numTestsNextMat^{(k)}_{\tnPoint\snPoint}$ is the number of tests observed along trace $(\tnPoint,\snPoint)$. \\
                Add $\numTestsNextMat^{(k)}$ to $\mathcal{N}$.
            }
            \For{$\numTestsNextMat^{(k)}\in\mathcal{N}$}{
                Randomly draw $\tnsnRates^{(k)}$ from $\truthdraws$.\\
                Simulate $\testResultNextMat^{(k)}$ according to $\numTestsNextMat^{(k)}$ and $\tnsnRates^{(k)}$, where $\testResultNextMat^{(k)}_{\tnPoint\snPoint}$ is the number of positive tests observed along trace $(\tnPoint,\snPoint)$.\\
                Add $\testResultNextMat^{(k)}$ to $\mathcal{Y}$.
            }
            \For{$\tnPoint\in\tnSet$, $\snPoint\in\snSet$}{
                Set $\numTestsNextMat^{\text{imp}}_{\tnPoint\snPoint}$ to the average number of tests for trace $(\tnPoint,\snPoint)$ across the members of $\mathcal{N}$, rounded to the nearest integer.\\
                Set $\testResultNextMat^{\text{imp}}_{\tnPoint\snPoint}$ to the average number of positive tests for trace $(\tnPoint,\snPoint)$ across the members of $\mathcal{Y}$, rounded to the nearest integer.
            }
            Return $\dataSetImp$, formed using $\numTestsNextMat^{\text{imp}}$ and $\testResultNextMat^{\text{imp}}$.
        }
        Set $\truthdraws, \datadraws:=\algparamone,\algparamtwo$ MCMC draws from $p(\tnsnRates|\dataSet)$. \\
        Get $\dataSetImp:=$ \texttt{BuildImportanceData}($\truthdraws$, $\allocVec$, $\sourcingMat$, $\algparamthree$).\\
        Set $\impdraws=\{\tnsnRates^{(i)}\}_{i=1}^{\algparamone}:=\algparamone$ MCMC draws from  $p(\tnsnRates|\dataSet,\dataSetImp)$.\\
        Get $\binomLikeMat:=$\texttt{BuildDatasimMatrix}($\impdraws$,  $\datadraws$, $\allocVec$, $\tnSet$, $\snSet$, $\sourcingMat$, $\algparamone$, $\algparamtwo$).\\
        \For{$j$ in $\{1,\dots,\algparamtwo\}$}{
            Set  $\bm{u}_j = \min_{\tnsnRatesEst} \bigg\{ \sum_{i=1}^{\algparamone}\binomLikeMat_{ij} L\big(\tnsnRatesEst, \tnsnRates^{(i)}\big) \frac{ p(\tnsnRates^{(i)}|\dataSet)}{p(\tnsnRates^{(i)}|\dataSet,\dataSetImp)} \bigg\}$
        }
        Return $\frac{1}{\algparamtwo}\sum_{j=1}^{\algparamtwo} \bm{u}_{j}$.
    \caption{Efficient estimation of $\expecOp\big[\estVal(\tnsnRatesEst(\allocVec),\tnsnRates)\big]$ using importance sampling}
    \label{alg:utilEstImpSamp}
\end{algorithm}
%%% END IMPORTANCE SAMPLING ALGORITHM %%%

\vspace{5pt}
Algorithm \ref{alg:utilEstImpSamp} produces an efficient expected loss estimate that reduces bias prevalent in high-dimensional settings. 
This algorithm uses the same inputs as Algorithm \ref{alg:utilApproxNodeFast}.
The key extension is the use of an additional set of MCMC draws of SFP-rate vectors to better capture the region of importance for the evaluated sampling plan.
An expected data set, $\dataSetImp$, is combined with known data $\dataSet$ for generating this additional MCMC set, $\impdraws$.
The matrix of binomial likelihoods for simulated data sets, $\binomLikeMat$, is generated using $\impdraws$, instead of $\truthdraws$ as in Algorithm \ref{alg:utilApproxNodeFast}.
The loss estimate under each simulated data set is adjusted by the importance weight before returning the estimate.

Algorithm \ref{alg:utilEstImpSamp} requires one additional parameter as compared with Algorithm \ref{alg:utilApproxNodeFast}:
$\algparamthree$, the number of data sets to generate for forming $\dataSetImp$.
Generating these data sets is computationally inexpensive.
We use $\algparamthree=5,000$ for the loss estimates of the case study, as evaluations showed this value yielded consistent resulting $\dataSetImp$ sets.
Calculating $\dataSetImp$ in this case requires only a few seconds on a standard laptop. 
Thus, this step contributes little to the overall calculation time as compared with MCMC generation, which requires a few minutes for our case study.

The principal disadvantage of Algorithm \ref{alg:utilEstImpSamp} is the additional MCMC generation required, which can be computationally expensive. 
This disadvantage compounds when evaluating multiple plans; the ability to evaluate multiple plans efficiently is a key goal of our work.
Whereas Algorithm \ref{alg:utilApproxNodeFast} can use a single MCMC set of SFP-rate vectors for evaluating any number of plans, Algorithm \ref{alg:utilEstImpSamp} requires additional MCMC generation for each plan.
However, the importance sampling of Algorithm \ref{alg:utilEstImpSamp} addresses a bias issue arising for Algorithm \ref{alg:utilApproxNodeFast} under sufficiently high dimensionality and a given maximal size of the MCMC set.

Future work will aim to discern when to apply each algorithm.
It is not clear how to use supply-chain size and computational limits to determine when importance sampling should be used: the bias of a loss estimate is a function of the set of MCMC draws, the number of nodes in the supply chain, and the underlying SFP rates.
One approach is to calibrate the amount of bias as a function of sampling budget by applying each algorithm at fixed sampling intervals, using a uniform allocation across test nodes. 
Any bias can then be estimated as the difference in loss estimate between the two algorithms.

\subsection{Utility approximations for case study}
Approximations of sampling plan utility depend significantly on the MCMC set of SFP-rate vectors used to generate the approximation.
Application of importance sampling requires new MCMC generation for approximations at each considered sampling budget.
Using different MCMC sets at each budget introduces variance.
We address this variance by re-running Algorithm \ref{alg:utilEstImpSamp} ten times with $\algparamone=75,000$, $\algparamtwo=2,000$, and $\algparamthree=5,000$.
We then compile the generated $\bm{u}_j$ loss values into a single loss vector; 95\% confidence interval values are calculated using this single vector.
\clearpage
%%% NEW OPTIMIZATION STEP SECTION
\section{Derivation of Bayes estimates for selected scores} \label{SuppMat:bayesEstimates}
Bayes estimates for the scores proposed in Sections \ref{subsubsec:objClass} and \ref{subsubsec:objAssess} can be derived analytically.
These estimates are posterior quantiles that can be readily obtained using a set of MCMC draws.
Using these analytical estimates instead of optimization software can improve computing time.
Propositions \ref{prop:bayes_absdiff} and \ref{prop:bayes_class} show Bayes estimates for these scores.
These propositions consider Bayes estimates of the SFP rate at a single node. 
Our proposed loss construction uses a summation of terms across nodes; thus, the loss pertaining to the estimate for any node can be considered separately from that of other nodes.

\vspace{15pt}
\begin{prop} \label{prop:bayes_absdiff}
    Under the Assessment score, the Bayes estimate for SFP rate $\tnsnRatesOneD$ at any node is $\tnsnRatesEstOneD^\star$, where
    \[
    \mathbbm{P}(\tnsnRatesOneD\leq\tnsnRatesEstOneD^\star)=\frac{\underEstWt}{1+\underEstWt}\,.
    \]
\end{prop}
\begin{proof}
Let $\tnsnRatesEstOneD$ be an estimate of SFP rate $\tnsnRatesOneD$, let $\dataSetNext$ be the PMS data available to choose $\tnsnRatesEstOneD$, and let $V(\tnsnRatesEstOneD)$ be the expected loss as a function of $\tnsnRatesEstOneD$.
The Bayes estimate $\tnsnRatesEstOneD^\star$ is the value of $\tnsnRatesEstOneD$ that minimizes $V(\tnsnRatesEstOneD)$.
Consider the expected loss:
    \begin{align*}
        V(\tnsnRatesEstOneD) &= \expecOp\big[(\tnsnRatesEstOneD-\tnsnRatesOneD)^{+}+\underEstWt(\tnsnRatesOneD-\tnsnRatesEstOneD)^{+}\big]\,.
    \end{align*}
The expression inside the expectation is convex, with a slope of 1 as $\tnsnRatesEstOneD$ moves to the right of $\tnsnRatesOneD$ and a slope of $\underEstWt$ as $\tnsnRatesEstOneD$ moves to the left of $\tnsnRatesOneD$.
Expectation preserves convexity; thus, $V$ is convex in $\tnsnRatesEstOneD$ and the value of $\tnsnRatesEstOneD$ where $V$ has a slope of zero is a minimizer of $V$.

Next, identify where $V$ has a slope of zero. 
The Leibniz rule indicates for functions $A$, $B$, and $\phi$ that
\[
    \left(\int_{A(x)}^{B(x)}\phi(x,t)dt \right)' = \int_{A(x)}^{B(x)}\phi'(x,t)dt+\phi\big(x,B(x)\big)B'(x)-\phi\big(x,A(x)\big)A'(x)\,.
\]
Expand the expected loss:
\begin{align*}
        V(\tnsnRatesEstOneD) &= \expecOp\big[(\tnsnRatesEstOneD-\tnsnRatesOneD)^{+}+\underEstWt(\tnsnRatesOneD-\tnsnRatesEstOneD)^{+}\big] \\
        &= \int_{-\infty}^{\tnsnRatesEstOneD} (\tnsnRatesEstOneD-\tnsnRatesOneD)p(\tnsnRatesOneD|\dataSetNext)d\tnsnRatesOneD + \underEstWt\int_{\tnsnRatesEstOneD}^{\infty} (\tnsnRatesOneD-\tnsnRatesEstOneD)p(\tnsnRatesOneD|\dataSetNext)d\tnsnRatesOneD\,.
\end{align*}
SFP rates lie in the $(0, 1)$ interval, so we can change the limits of integration:
\begin{align*}
        V(\tnsnRatesEstOneD) &= \int_{0}^{\tnsnRatesEstOneD} (\tnsnRatesEstOneD-\tnsnRatesOneD)p(\tnsnRatesOneD|\dataSetNext)d\tnsnRatesOneD + \underEstWt\int_{\tnsnRatesEstOneD}^{1} (\tnsnRatesOneD-\tnsnRatesEstOneD)p(\tnsnRatesOneD|\dataSetNext)d\tnsnRatesOneD\,.
\end{align*}
Apply the Leibniz rule to obtain the derivative of $V$ with respect to $\tnsnRatesEstOneD$:
    \begin{align*}
        V'(\tnsnRatesEstOneD) &= \int_{0}^{\tnsnRatesEstOneD} p(\tnsnRatesOneD|\dataSetNext)d\tnsnRatesOneD  - \underEstWt \int_{\tnsnRatesEstOneD}^{1}p(\tnsnRatesOneD|\dataSetNext)d\tnsnRatesOneD\,.
    \end{align*}
Setting $V'(\tnsnRatesEstOneD)$ to zero and solving for $\tnsnRatesEstOneD$ produces Bayes estimate $\tnsnRatesEstOneD^\star$:
    \begin{align*}
        \int_{0}^{\tnsnRatesEstOneD^\star} p(\tnsnRatesOneD|\dataSetNext)d\tnsnRatesOneD - \underEstWt \int_{\tnsnRatesEstOneD^\star}^{1} p(\tnsnRatesOneD|\dataSetNext)d\tnsnRatesOneD &=0 \\
        \mathbbm{P}(\tnsnRatesOneD\leq \tnsnRatesEstOneD^\star) &= \underEstWt\big[1- \mathbbm{P}(\tnsnRatesOneD\leq \tnsnRatesEstOneD^\star)\big] \\
        \mathbbm{P}(\tnsnRatesOneD\leq \tnsnRatesEstOneD^\star) &= \frac{\underEstWt}{1+\underEstWt}
    \end{align*}
\end{proof}

\begin{prop} \label{prop:bayes_class}
    Under the Classification score using threshold $\threshold$, the Bayes estimate for the SFP rate at a single node is 
    \[
    \begin{cases}
         1 \text{ if } \mathbbm{P}(\tnsnRatesOneD\leq \threshold) \leq \frac{\underEstWt}{1+\underEstWt} \text{ and} \\
         0 \text{ o.w.}
     \end{cases}
    \]
\end{prop}
\begin{proof}
    In the Classification setting, a ``1'' denotes nodes classified as significant SFP sources, and a ``0'' denotes nodes classified as non-significant sources. 
    Consider the expected loss, $V(\cdot)$, when classifying a node as 1 or 0:
    \begin{align*}
        V(1) &= \expecOp\Big[\big(1-C(\tnsnRatesOneD)\big)^{+}+\underEstWt\big(C(\tnsnRatesOneD)-1\big)^{+}\Big] \\
        &= \int_{0}^{\threshold} (1) p(\tnsnRatesOneD|\dataSetNext)d\tnsnRatesOneD  + \underEstWt \int_{\threshold}^{1} (0) p(\tnsnRatesOneD|\dataSetNext)d\tnsnRatesOneD\ \\
        &= \int_{0}^{\threshold} p(\tnsnRatesOneD|\dataSetNext)d\tnsnRatesOneD\,,
    \end{align*}
    \begin{align*}
        V(0) &= \expecOp\Big[\big(0-C(\tnsnRatesOneD)\big)^{+}+\underEstWt\big(C(\tnsnRatesOneD)-0\big)^{+}\Big] \\
        &= \int_{0}^{\threshold} (0) p(\tnsnRatesOneD|\dataSetNext)d\tnsnRatesOneD  + \underEstWt \int_{\threshold}^{1} (1) p(\tnsnRatesOneD|\dataSetNext)d\tnsnRatesOneD\ \\
        &= \underEstWt\int_{\threshold}^{1} p(\tnsnRatesOneD|\dataSetNext)d\tnsnRatesOneD\,.
    \end{align*}
    We classify the node as a significant SFP source when the expected loss of doing so is less than the expected loss of the alternative, i.e., when $V(1)\leq V(0)$.
    Consider when $V(1)\leq V(0)$:
    \begin{align*}
         \int_{0}^{\threshold} p(\tnsnRatesOneD|\dataSetNext)d\tnsnRatesOneD &\leq \underEstWt\int_{\threshold}^{1} p(\tnsnRatesOneD|\dataSetNext)d\tnsnRatesOneD \\
         \mathbbm{P}(\tnsnRatesOneD\leq\threshold) &\leq \underEstWt\big[1-\mathbbm{P}(\tnsnRatesOneD\leq\threshold)\big] \\
         \mathbbm{P}(\tnsnRatesOneD\leq\threshold) &\leq \frac{\underEstWt}{1+\underEstWt}\,.
    \end{align*}
    Thus, we classify the node as a significant SFP source when $\mathbbm{P}(\tnsnRatesOneD\leq\threshold) \leq \frac{\underEstWt}{1+\underEstWt}$.
\end{proof}

Observe that these propositions ignore the weight of an SFP rate, $W(\tnsnRates)$. 
Our implementation of utility estimation accounts for the weight alongside the binomial likelihood of the data simulation matrix: both the weight and the likelihood are functions of the SFP rate and not the estimate. 
\clearpage
\section{Inference preparation for case study}
\label{SuppMat:caseStudyPrep}

This appendix describes preparation steps from our case study for using a utility approach to form sampling plans.

We examine eliciting priors on SFP rates.
An appropriate prior has two properties: right skewness to capture the possibility of high SFP rates, and increased variance as assessed risk increases.
The SFP risk assessment process of \citet{nkansah2018} establishes seven risk categories through appraisals of factors like population density and poverty levels.
Adapting these risk assessments yields priors with the desired properties, while avoiding more extensive prior elicitation:
\begin{equation}
    p(\tnsnRates) \propto \prod\limits_{\tnsnPoint\in \tnsnSet } \exp\left\{-\frac{1}{2}\left[\frac{\logitfunc(\tnsnRates_\tnsnPoint)-\logitfunc(\mu_{q(\tnsnPoint)})}{\priorVar}\right]^2\right\} \,.
\end{equation}
Function $h$ is the logit function; the prior is normal in the logit space.
Parameter $\mu_{q(\tnsnPoint)}$ is a regulator-designated SFP-rate median for node $\tnsnPoint$ with assessed SFP risk $q(\tnsnPoint)$.
Our field-user co-authors established $\mu$ values of 1\%, 2\%, 5\%, 10\%, 15\%, 20\% and 25\% for the seven SFP risk levels of \citeauthor{nkansah2018}.
Parameter $\priorVar$ is a variance parameter, used for all assessed SFP risk levels.
We used a $\priorVar$ of $2$ for the case study, which produced densities matching field-user expectations.
These densities are shown in Figure \ref{fig:SFPriskDensities}.
For example, a ``Very Low'' assessment corresponds to a prior 90\% interval of $[1\%, 38\%]$, while a ``High'' assessment corresponds to an interval of $[7\%, 72\%]$.
Use of the logit space provides the desired right skew in the $(0,1)$ space, as well as increased variance as the assessed SFP risk increases.
This variance pattern is achieved through only a single 
choice of $\priorVar$; using only one $\priorVar$ value is desirable as it concisely captures overall confidence in risk assessments.

%%% FIGURE FOR CASE STUDY PRIOR DENSITIES %%%
\begin{figure}[ht]
    \centering
    \includegraphics[width=0.8\textwidth]{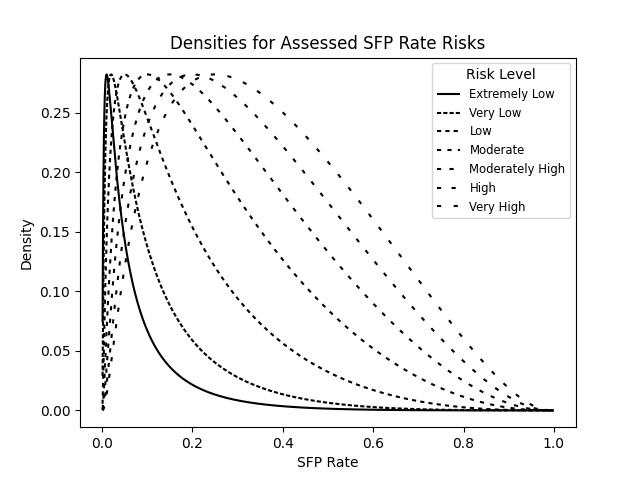}
    \caption[Prior densities for \citet{nkansah2018} risk designations]{Prior densities at SFP rates for seven designations of SFP risk, as introduced in \citet{nkansah2018}, using variance $\priorVar=2$.}
    \label{fig:SFPriskDensities}
\end{figure}
%%% END FIGURE %%%

\subsection*{Estimating sourcing probability}
The sourcing-probability matrix for the supply chain, $\sourcingMat$, is a crucial element of understanding what data may result from a proposed sampling plan.
National sourcing records, aggregated records from sub-national or outlet echelons, and track-and-trace records can all be used to form $\sourcingMat$.
Our case study uses supply-chain observations from a previous PMS iteration to estimate $\sourcingMat$, and we use bootstrap samples of these observations to form sourcing probabilities for untested nodes in the all-provinces case.
We used 44 independent draws from the first iteration's data set to form the sourcing probabilities for each test node; 44 was the average number of samples per test node in the first iteration.

The uncertainty associated with an estimated $\sourcingMat$ may be incorporated into analysis through similar bootstrap samples that are integrated with MCMC draws of the SFP-rate vector, $\tnsnRates$.
In the efficient estimation algorithm (see Appendix \ref{SuppMat:calcUtil}), each member of the set of SFP-rate draws would be coupled with a bootstrap sample of $\sourcingMat$ during the data-generation step.
The utility estimate would thus be reflective of the uncertainty in both $\tnsnRates$ and $\sourcingMat$.
Further discussion on sourcing probability is contained in
\if0\blind{\citet{wickett2023} and \citet{wickettdissertation}.}\fi
\if1\blind{\citet{wickettdissertationblind}.}\fi

\end{appendices}

%TC:endignore

\end{document}